\PassOptionsToPackage{unicode}{hyperref}
\PassOptionsToPackage{hyphens}{url}
\PassOptionsToPackage{dvipsnames,svgnames,x11names}{xcolor}
\documentclass[
]{article}
\usepackage{xcolor}
\usepackage[top=10pc,bottom=10pc,left=11pc,right=11pc,heightrounded,top=1in,bottom=1in,left=1in,right=1in]{geometry}
\usepackage{amsmath,amssymb}
\usepackage{iftex}
\ifPDFTeX
  \usepackage[T1]{fontenc}
  \usepackage[utf8]{inputenc}
  \usepackage{textcomp} % provide euro and other symbols
\else % if luatex or xetex
  \usepackage{unicode-math} % this also loads fontspec
  \defaultfontfeatures{Scale=MatchLowercase}
  \defaultfontfeatures[\rmfamily]{Ligatures=TeX,Scale=1}
\fi
\usepackage{lmodern}
\ifPDFTeX\else
\fi
\IfFileExists{upquote.sty}{\usepackage{upquote}}{}
\IfFileExists{microtype.sty}{% use microtype if available
  \usepackage[]{microtype}
  \UseMicrotypeSet[protrusion]{basicmath} % disable protrusion for tt fonts
}{}
\usepackage{setspace}

\usepackage{longtable,booktabs,array}
\usepackage{calc} % for calculating minipage widths
\usepackage{etoolbox}
\makeatletter
\patchcmd\longtable{\par}{\if@noskipsec\mbox{}\fi\par}{}{}
\makeatother
\IfFileExists{footnotehyper.sty}{\usepackage{footnotehyper}}{\usepackage{footnote}}
\makesavenoteenv{longtable}
\usepackage{graphicx}
\makeatletter
\newsavebox\pandoc@box
\newcommand*\pandocbounded[1]{% scales image to fit in text height/width
  \sbox\pandoc@box{#1}%
  \Gscale@div\@tempa{\textheight}{\dimexpr\ht\pandoc@box+\dp\pandoc@box\relax}%
  \Gscale@div\@tempb{\linewidth}{\wd\pandoc@box}%
  \ifdim\@tempb\p@<\@tempa\p@\let\@tempa\@tempb\fi% select the smaller of both
  \ifdim\@tempa\p@<\p@\scalebox{\@tempa}{\usebox\pandoc@box}%
  \else\usebox{\pandoc@box}%
  \fi%
}
\def\fps@figure{htbp}
\makeatother

\NewDocumentCommand\citeproctext{}{}
\NewDocumentCommand\citeproc{mm}{%
  \begingroup\def\citeproctext{#2}\cite{#1}\endgroup}
\makeatletter
 \let\@cite@ofmt\@firstofone
 \def\@biblabel#1{}
 \def\@cite#1#2{{#1\if@tempswa , #2\fi}}
\makeatother
\newlength{\cslhangindent}
\newlength{\csllabelwidth}
\newenvironment{CSLReferences}[2] % #1 hanging-indent, #2 entry-spacing
 {\begin{list}{}{%
  \setlength{\itemindent}{0pt}
  \setlength{\leftmargin}{0pt}
  \setlength{\parsep}{0pt}
  \ifodd #1
   \setlength{\leftmargin}{\cslhangindent}
   \setlength{\itemindent}{-1\cslhangindent}
  \fi
  \setlength{\itemsep}{#2\baselineskip}}}
 {\end{list}}
\usepackage{calc}

\providecommand{\tightlist}{%
  \setlength{\itemsep}{0pt}\setlength{\parskip}{0pt}}

\usepackage[all,defaultlines=2]{nowidow}

\usepackage{enumitem}
\setlist[1]{labelindent=\parindent}
\setlist[itemize]{leftmargin=*}
\setlist[enumerate]{leftmargin=*}

\setlist[description]{style=unboxed}

\usepackage[titles]{tocloft}

\AtBeginEnvironment{longtable}{\setlist[itemize]{nosep, wide=0pt, leftmargin=*, before=\vspace*{-\baselineskip}, after=\vspace*{-\baselineskip}}}

\usepackage{setspace}

\usepackage[overload]{textcase}
\usepackage[runin]{abstract}

\abslabeldelim{\quad}
\newenvironment{keywords}
{\vskip -3em \hspace{\parindent}\small\sffamily{\sffamily\footnotesize\bfseries\MakeUppercase{Keywords}}\quad}
{\vskip 3em}

\usepackage{titling}
\pretitle{\par\vskip 5em \begin{flushleft}\LARGE\sffamily\bfseries}
\posttitle{\par\end{flushleft}\vskip 0.75em}

\renewcommand{\and}{\end{tabular} \hskip 3em \begin{tabular}[t]{@{\hspace{0em}}l@{}}}
\preauthor{\begin{flushleft}
           \lineskip 1.5em 
           \begin{tabular}[t]{@{\hspace{0em}}l@{}}}
\postauthor{\end{tabular}\par\end{flushleft}}

\predate{}
\postdate{}

\newcommand{\published}[1]{%
   \gdef\puB{#1}}
   \newcommand{\puB}{}

\usepackage{titlesec}
\titleformat*{\section}{\Large\sffamily\bfseries\raggedright}
\titleformat*{\subsection}{\large\sffamily\bfseries\raggedright}
\titleformat*{\subsubsection}{\normalsize\sffamily\bfseries\raggedright}
\titleformat*{\paragraph}{\small\sffamily\bfseries\raggedright}

\titlespacing*{\section}{0em}{2em}{0.1em}
\titlespacing*{\subsection}{0em}{1.25em}{0.1em}
\titlespacing*{\subsubsection}{0em}{0.75em}{0em}

\usepackage[bottom, flushmargin]{footmisc}

\usepackage[font={small,sf}, labelfont={small,sf,bf}]{caption}

\usepackage{pdflscape}
\newcommand{\blandscape}{\begin{landscape}}
\newcommand{\elandscape}{\end{landscape}}

\usepackage{bbm}
\let\origmathbb\mathbb
\renewcommand{\mathbb}[1]{\ifnum\pdfstrcmp{#1}{1}=0 \mathbbm{1}\else\origmathbb{#1}\fi}
\usepackage{xcolor}
\definecolor{coloraccent}{HTML}{107895}
\usepackage[section]{placeins}
\usepackage{float}
\usepackage{tabularray}
\usepackage[normalem]{ulem}
\usepackage{graphicx}
\usepackage{rotating}
\UseTblrLibrary{siunitx}
\NewTableCommand{\tinytableDefineColor}[3]{\definecolor{#1}{#2}{#3}}

\makeatletter
\@ifpackageloaded{caption}{}{\usepackage{caption}}
\AtBeginDocument{%
\ifdefined\contentsname
  \renewcommand*\contentsname{Table of contents}
\else
  \newcommand\contentsname{Table of contents}
\fi
\ifdefined\listfigurename
  \renewcommand*\listfigurename{List of Figures}
\else
  \newcommand\listfigurename{List of Figures}
\fi
\ifdefined\listtablename
  \renewcommand*\listtablename{List of Tables}
\else
  \newcommand\listtablename{List of Tables}
\fi
\ifdefined\figurename
  \renewcommand*\figurename{Figure}
\else
  \newcommand\figurename{Figure}
\fi
\ifdefined\tablename
  \renewcommand*\tablename{Table}
\else
  \newcommand\tablename{Table}
\fi
}
\@ifpackageloaded{float}{}{\usepackage{float}}
\floatstyle{ruled}
\@ifundefined{c@chapter}{\newfloat{codelisting}{h}{lop}}{\newfloat{codelisting}{h}{lop}[chapter]}
\floatname{codelisting}{Listing}

\usepackage{amsthm}
\theoremstyle{plain}
\newtheorem{proposition}{Proposition}[section]
\theoremstyle{definition}
\newtheorem{definition}{Definition}[section]
\theoremstyle{remark}
\AtBeginDocument{}

\makeatother
\makeatletter
\@ifpackageloaded{caption}{}{\usepackage{caption}}
\@ifpackageloaded{subcaption}{}{\usepackage{subcaption}}
\makeatother
\usepackage{bookmark}
\IfFileExists{xurl.sty}{\usepackage{xurl}}{} % add URL line breaks if available
\hypersetup{
  pdftitle={A Quantitative Model of Non-Marriage and Fertility},
  pdfauthor={Kazuharu Yanagimoto},
  pdfkeywords={Marriage, Fertility, Bargaining, Leisure Technology},
  colorlinks=true,
  linkcolor={coloraccent},
  filecolor={Maroon},
  citecolor={coloraccent},
  urlcolor={coloraccent},
  pdfcreator={LaTeX via pandoc}}

\usepackage{orcidlink}  % Create automatic ORCID icons/links

\title{A Quantitative Model of Non-Marriage and Fertility\thanks{I am
grateful to my advisor Nezih Guner for his continued guidance and
support. I thank Pedro Mira, Tom Zohar, and seminar participants at
CEMFI Macro Reading Group, for their useful suggestions. I acknowledge
financial support from the Maria de Maeztu Unit of Excellence CEMFI
MDM-2016-0684, funded by MCIN/AEI/10.13039/501100011033 and CEMFI, and
from JSPS KAKENHI Grant Number 25K23111.}}

\usepackage{etoolbox}
\makeatletter
\providecommand{\subtitle}[1]{% add subtitle to \maketitle
  \apptocmd{\@title}{\par {\vskip 0.25em \large #1 \par}}{}{}
}
\makeatother
\subtitle{Bargaining over Leisure}

\author{
{\large Kazuharu Yanagimoto~\orcidlink{0009-0007-1967-8304}}%
 \\%
Kobe University \\%
{\footnotesize \url{yanagimoto@econ.kobe-u.ac.jp}} \and
}

\date{}

\usepackage{url}
\begin{document}
% ---------------
% TITLE SECTION
% ---------------
\published{\textbf{Monday, March 16, 2026}}

\maketitle

\begin{abstract}
This paper introduces a new factor contributing to the decline in
marriage and fertility: the growth of leisure technology. Over recent
decades, high-income countries have experienced two notable shifts in
household and family dynamics. First, there has been a significant
decline in marriage rates and fertility. Second, time has increasingly
been allocated to leisure activities. This paper presents a unified
model of marriage and fertility, incorporating intra-household
bargaining dynamics. The model, calibrated using data from Japan between
2019 and 2023, is employed to assess the impact of leisure technology
growth on marriage and fertility during 2005-2009. The findings
highlight that leisure technology growth makes single life relatively
more appealing compared to marriage and parenthood. The model explains
21.1\% of the decline in marriage and 73.1\% of the decrease in
fertility.
\end{abstract}
\vskip 3em

\begin{keywords}
\def\sep{;\ }
Marriage\sep Fertility\sep Bargaining\sep 
Leisure Technology
\end{keywords}
% -------------------
% END TITLE SECTION
% -------------------

\setstretch{1.25}
\vskip -3em \hspace{\parindent}{\sffamily\footnotesize\bfseries\MakeUppercase{JEL Codes}}\quad {\sffamily\small J12; J13; D13}\vskip 3em

\newpage

\section{Introduction}\label{sec-intro}

Most of the developed countries have been facing a decline in marriage
(or partnerships) and fertility rates. The decline in fertility is often
viewed as a major policy concern, and many governments implement
family-friendly policies to encourage childbearing. A substantial body
of labor and macroeconomics literature investigates the mechanism behind
fertility decline and the design of optimal policies. However, the
connection between household formation and fertility, the factors
driving the decline in marriage rates, and their potential influence on
fertility decisions are not fully understood.

This paper proposes a new driver of the decline in marriage and
fertility: the improvements in leisure technology. Leisure technology,
such as video games and social media, has been dramatically improved in
the last decade, and the value of leisure has increased. One potential
impact of the growth of leisure technology could be the increased value
of being single. Singles can enjoy their hobbies and leisure time
without any constraints. At the same time, married couples have to
compromise their leisure time with their partners and children, lowering
incentives for marriage and partnerships.

In this paper, I study the impact of leisure technology growth on the
marriage and fertility decline in Japan. Several factors make Japan an
ideal case for studying low fertility. First, marriage and childbirth
are still strongly associated in Japan. According to OECD data, the
share of births outside of marriage in Japan was 2.4\% in 2020, the
lowest among the OECD countries. Hence, to understand low fertility in
Japan, it is crucial to study the entry into marriage and how married
couples choose the number of children. Second, the fertility rate in
Japan has been declining since 2000, and the share of those who have
never married has steadily increased over the past decades. Since
childbirth outside of marriage is very uncommon in Japan, the increase
in the share of never-married is likely to be a significant contributor
to the decline in fertility.

To explore the impact of leisure technology on the decline in marriage
and fertility in Japan, I begin by documenting key patterns, including
the trends in marriage and fertility rates, the reasons for remaining
single, and the time allocation differences between singles and married
couples. First, survey data reveals that 22.4\% of men and 24.5\% of
women choose to remain single to pursue hobbies, while 26.6\% of men and
31\% of women prefer the freedom of being single. This indicates that
the desire and opportunity to enjoy more leisure time as a single person
can discourage marriage, as singles believe their leisure time will
decrease after marriage. I refer to this as the ``marriage penalty on
leisure.'' Second, panel data analysis shows that wives experience a
significantly greater decline in leisure time than their husbands, which
I term the ``child penalty on leisure.'' Finally, I document that the
distribution of leisure time between married couples is influenced by
their relative wages, with the lower-wage partner having less leisure
time. This finding aligns with household bargaining, where the
higher-earning partner enjoys more leisure time, contributing to the
marriage penalty on leisure.

To quantify the impact of leisure technology on the decline in marriage
and fertility, I develop a model of time allocation and household
formation. This model includes single and married individuals, each with
a stochastic life cycle. Individuals differ in their labor market
productivity. In each period, singles are randomly matched with other
singles, and marriage occurs if both parties agree. Married individuals
must decide how to allocate their time among work, household work, and
leisure, as well as when and whether to have children. While children
bring utility, they also increase the household workload for parents.
Household work is modeled as a time requirement that must be met by
aggregating time inputs from husbands and wives. The parameters of this
aggregation are allowed to change over time to reflect changes in the
social norms, which become more egalitarian.

Marriage has economic value due to resource pooling and a random utility
value that reflects the match quality between partners. Although married
couples have more resources, their time allocation is more constrained.
A couple jointly decides how to allocate time and how many children to
have, with bargaining power determined by their relative wages.
Individuals with lower wages have less leisure time and experience a
more significant reduction in leisure time when they have children.
Households get utility from consumption, children, and leisure. The
utility weight of leisure is allowed to change over time to capture
changes in the leisure technology. When deciding whether to marry,
individuals consider the value of remaining single, including potential
opportunities to meet other partners. The distribution of potential
partners is endogenous and is determined by the decisions of all
individuals.

I estimate the parameters of the structural model using data for the
2019-2023 period. The model generates the observed pattern of household
time allocation, marriage, and fertility rates. It also captures the
heterogeneous marriage rate by earnings and the child penalty on
leisure, which are not targeted in the calibration.

I then apply the model to examine the factors driving the decline in
marriage and fertility in Japan over recent decades, with a particular
focus on the impact of leisure technology growth. In addition to leisure
technology, I also consider the effects of rising female wages and
changes in gender roles within households. These two factors have also
undergone significant shifts in the past decade and have been
extensively studied in the literature regarding their influence on
marriage and fertility rates.

The model successfully accounts for a substantial portion of the
observed decline in marriage and fertility over the last decade,
explaining 21.1\% of the decline in marriage and 73.1\% of the decrease
in fertility. Decomposition analysis reveals that the growth of leisure
technology has been the most significant factor contributing to the
decline in fertility, as it increases the relative value of leisure
compared to having children.

\textbf{Related Literature} My first contribution is to introduce the
leisure technology growth as a new driver of the marriage and fertility
decline The impact of leisure technology growth on the labor supply and
the value of leisure has been studied by others. Kopecky
(\citeproc{ref-kopecky2011}{2011}) builds a model of endogenous
retirement with leisure production and shows that the decline in the
price of leisure goods makes retirement more attractive. Aguiar et al.
(\citeproc{ref-aguiar2021}{2021}) show that the leisure technology
growth in computer games can explain the 2\% decline in the labor supply
of young men in the US in the last decade. Kopytov et al.
(\citeproc{ref-kopytov2023}{2023}) also show the decline in prices of
recreational activities can explain a large proportion of the decrease
in the labor working hours in some OECD countries.\footnote{ For the
  United States, González-Chapela
  (\citeproc{ref-gonzalez-chapela2007}{2007}) studied the negative
  impact of the decline in the price of recreational goods on hours in
  the labor market from 1976 to 93 and Vandenbroucke
  (\citeproc{ref-vandenbroucke2009}{2009}) from 1900 to the 1950s.}
However, little is known about the impact of leisure technology growth
on the marriage and fertility decline. By incorporating the growth of
leisure technology into the model, I show that leisure technology growth
can explain a significant proportion of the decline in marriage and
fertility in the last decade.

This paper also contributes to the labor and macroeconomics literature
that studies changes in marriage and fertility in high-income
countries.\footnote{ See Greenwood et al.
  (\citeproc{ref-greenwood2023}{2023}) and Doepke et al.
  (\citeproc{ref-doepke2023}{2023}) for comprehensive reviews of this
  literature.} As an early work, Ahn and Mira
(\citeproc{ref-ahn2002}{2002}) points out the negative correlation
between fertility and female labor force participation in high-income
countries. Female earnings and labor market arrangements that affect
female labor supply have been studied as a key factor in the recent
marriage and fertility decline (Santos and Weiss
(\citeproc{ref-santos2016}{2016}); Adda et al.
(\citeproc{ref-adda2017}{2017}); Blasutto
(\citeproc{ref-blasutto2023}{2023}); Guner et al.
(\citeproc{ref-guner2023}{2023}); Cruces
(\citeproc{ref-cruces2024}{2024})). In this literature, Greenwood et al.
(\citeproc{ref-greenwood2016}{2016}) build a dynamic marriage model with
growth in home production technology. While their model explains the
decline in marriage rates in the US by the decrease in the price of home
production inputs, the model does not include an endogenous fertility
decision or leisure choice. Baudin et al.
(\citeproc{ref-baudin2015}{2015}) build a model of endogenous marriage
and fertility. However, their model is static and does not consider a
dynamic household formation.\footnote{ Myong et al.
  (\citeproc{ref-myong2021}{2021}) extend the model of Baudin et al.
  (\citeproc{ref-baudin2015}{2015}) by incorporating the social norms
  and explain the marriage and fertility decline in South Korea.}

Another related literature focuses on the importance of the bargaining
power and intra-household decision-making.\footnote{ Basu
  (\citeproc{ref-basu2006}{2006}) and Iyigun and Walsh
  (\citeproc{ref-iyigun2007}{2007}) provide theoretical models of
  endogenous bargaining power.} Knowles
(\citeproc{ref-knowles2013}{2013}) studies intra-household bargaining
and labor supply and shows that gender asymmetry in bargaining power
explains the time trends of female labor supply in the US from 1970 to
2001. Burda et al. (\citeproc{ref-burda2013}{2013}) document a negative
relationship between GDP per capita and gender differences in total work
and emphasize the role of social norms and intra-household bargaining.
To study the relationship between bargaining power and fertility, Doepke
and Kindermann (\citeproc{ref-doepke2019}{2019}) build a Nash-bargaining
model of fertility. While these models highlight the importance of
relative bargaining power in explaining intra-household decision-making,
the partnership formation with the bargaining after the marriage is
absent. The current model endogenizes the marriage given the expectation
of household decision-making after the marriage.

Finally, this paper deepens the economic understanding of the decline in
marriage and fertility in Japan. Kitao and Nakakuni
(\citeproc{ref-kitao2024b}{2024}) build a static model of marriage and
fertility and show that one of the main drivers of the marriage and
fertility decline in Japan from 1970 to 2020 is the change in home
production technology and the increase in time and financial costs of
childcare. While this paper captures the trends in marriage and
fertility rates in the past decades, it does not consider the impact of
leisure technology growth and gender asymmetries in bargaining power.
Lise and Yamada (\citeproc{ref-lise2019}{2019}) study the collective
model of intra-household allocation and welfare analysis in Japan, and
Guo and Xie (\citeproc{ref-guo2024}{2024}) extend it by incorporating
the arrival of the first child. While these papers focus on
intra-household decision-making, the current paper investigates the
marriage market and family formation, given the expectation of
intra-household bargaining.

The rest of the paper is organized as follows. Section~\ref{sec-facts}
documents the stylized facts of the decline in marriage and fertility.
Section~\ref{sec-model} presents the model of time allocation and
household formation. Section~\ref{sec-calib} describes the calibration
strategy. Section~\ref{sec-baseline} estimates the model parameters
using the census and the household survey. Section~\ref{sec-past} shows
the model's ability to reproduce the observed pattern of the marriage
and fertility rates. Section~\ref{sec-concl} concludes the paper.

\section{Facts}\label{sec-facts}

This section investigates the reasons why people do not get married and
provides stylized facts to help the model construction. The results
point out the importance of leisure time in the marriage decision and
suggest that intra-household bargaining power might disincentivize
people from marriage and childbirth.

\subsection{Marriage and Fertility
Trends}\label{marriage-and-fertility-trends}

The fertility rate in Japan has been declining since 2000. Panel (a) of
Figure~\ref{fig-nmarried-fertility} shows the fertility rate of women at
age 45. The fertility rate started to decline from 2000 and reached 1.47
in 2022. During the same period, the share of those who have never
married has dramatically increased in the past decades in Japan. Panel
(b) of Figure~\ref{fig-nmarried-fertility} shows the share of those who
were never married at age 45-54. For men, the share of those who are
never married has gradually increased from 1980 and reached 25.8\% in
2020. For women, the share started to increase from 2000 when the
fertility rate began to decline and reached 16.4\% in 2020.

\begin{figure}

\centering{

\includegraphics[width=0.9\linewidth,height=\textheight,keepaspectratio]{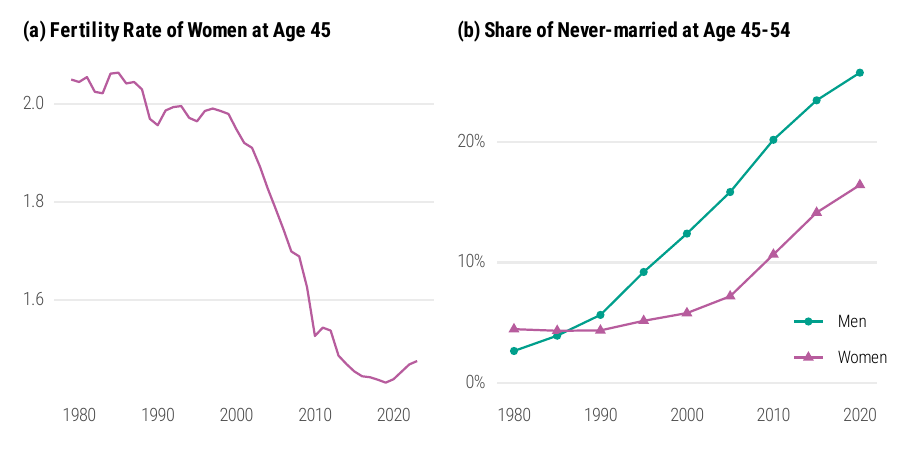}

}

\caption{\label{fig-nmarried-fertility}\textbf{Marriage and Fertility in
the Past Decades in Japan}. Panel (a) shows the fertility rate of women
at age 45. The data is from the Human Fertility Database. Panel (b)
shows the share of those who have never married at age 45-54. The data
is from the Japanese Censuses.}

\end{figure}%

\subsection{Why People Do Not Marry}\label{why-people-do-not-marry}

What is the main reason why people do not get married? The National
Fertility Survey (NFS), which investigates the situation and issues
regarding marriage, childbirth, and child-rearing every five years,
provides information on the main reasons people do not get married.
Figure~\ref{fig-reason-single} shows the top 5 reasons in
2021.\footnote{ I plot the time trends of these reasons from 1992 to
  2021 in Figure~\ref{fig-reason-single-ts}. The share of the top
  reasons does not change much over time. However, given that the share
  of never-married people is increasing and the survey sample is limited
  to singles, the number of people who agree with these reasons might be
  growing.} The sample is restricted to singles aged 25-34.

\begin{figure}

\centering{

\includegraphics[width=0.9\linewidth,height=\textheight,keepaspectratio]{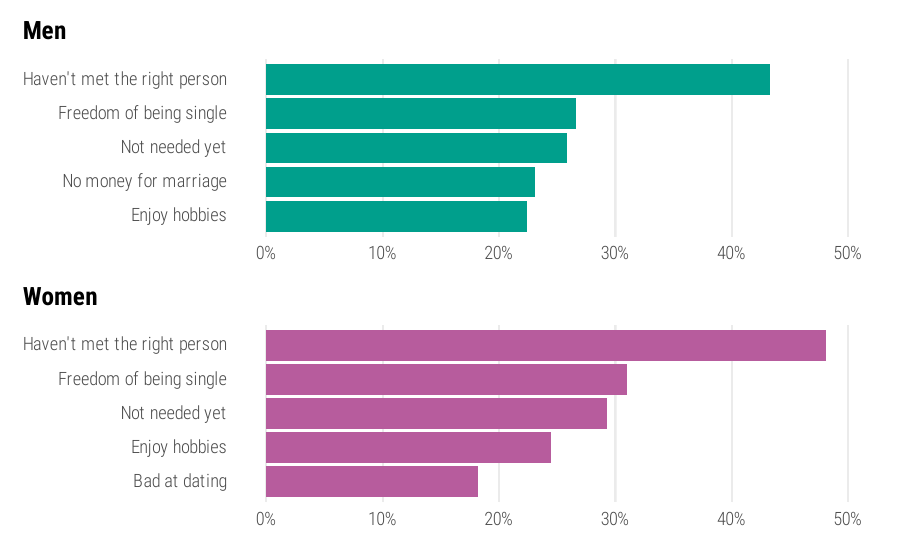}

}

\caption{\label{fig-reason-single}\textbf{Top 5 Reasons Why Do not Get
Married}. The share of the top 5 reasons why do not get married. The
data is from the National Fertility Survey, and the sample is restricted
to the age group 25-34 in 2021.}

\end{figure}%

For both men and women, the top reason is ``I haven't met the right
person'', which implies their difficulty in finding a partner with whom
they spend a better time than being single. The third reason for both,
``Not needed yet'', can also be interpreted in this context. They might
expect their life will not improve with their potential partner. These
elements will be represented in the model as a matching process with a
potential partner and their marriage decision based on their expected
value of being single and married.

``Freedom of being single'' and ``Enjoy hobbies'' are also important
reasons for both men and women. It implies that they expect marriage to
restrict their lives and not allow them to enjoy their hobbies enough.
In other words, once they get married, they cannot live as well as they
did when they were single. In the model, this will be represented as a
couple's joint decision of time allocation and the bargaining power
between them. If they get married, especially with a partner with a
higher bargaining power, they may have less leisure time, and their
utility could be lower than if they remained single.

\subsection{Data and Samples}\label{sec-data-samples}

The primary data source for the empirical analysis is the Japanese
Household Panel Survey (JHPS). JHPS is panel data that started in 2004,
with the most recent data wave from 2023. The original 2004 sample was
nationwide and contained 4000 households and 7000 individuals older than
20. Additional samples were added in 2007 (1400 individuals), 2009 (4000
individuals), and 2012 (1000 individuals). The survey has information on
earnings, working hours, and other labor market outcomes, as well as on
family structure, fertility, and other demographic variables. For
details on the data, see Section~\ref{sec-app-jhps}.

The survey also has information on hours spent on housework and
childcare, and I define the hours spent on leisure as the residuals of
the total time budget. For the following discussion, I will use the
terms ``working hours'' and ``hours worked'' to refer to the sum of
hours spent on the market and commuting. ``Domestic labor'' refers to
the sum of hours spent on housework and childcare, and ``leisure''
refers to the residual. I assume the total time budget is 16 hours per
day, and \(16 \times 7 = 112\) hours per week. Hence, the weekly leisure
hours are calculated as the total time budget (112 hours) minus the sum
of working hours and domestic labor per week.

The sample is restricted to people aged 25-54 in the period 2005-2023.
The data in 2004 is not used because the domestic labor data is not
available. The singles sample is restricted to those with a job, a
positive leisure time, and no children. The sample of married couples is
also restricted to those with a positive leisure time. However, it
includes non-working individuals.

\subsection{Child Penalty on Leisure}\label{child-penalty-on-leisure}

Another important motivation for marriage is having children. On the
other hand, raising a child is time-consuming and reduces parents'
leisure time. Hence, the reasons why people do not get married,
``Freedom of being single'' and ``Enjoy hobbies'', might be related to
childbearing and its implications for time allocations within marriage.

I conduct an event study analysis to investigate the impact of
childbirth on leisure time. The specification is given by

\begin{equation}\protect\phantomsection\label{eq-event-study}{
y_{it} = \alpha_i + \lambda_t + \sum_{q \neq -2, -\infty} \beta_q \mathbb{1}\{C_i + q = t\} + \varepsilon_{it}.
}\end{equation}

The \(y_{it}\) is the time allocation of individual \(i\) at time \(t\),
\(\alpha_i\) is the individual fixed effect, \(\lambda_t\) is the time
fixed effect, and \(\varepsilon_{it}\) is the error term. The \(C_i\) is
the time of the first childbirth of individual \(i\), so the \(k\)
represents the relative years to the first childbirth. The sample
consists of the individuals who had their first childbirth during the
sample period and the individuals who did not have a child in the sample
period. For the individuals who did not have a child, the \(C_i\) is set
to \(\infty\), and the \(k\) is set to \(-\infty\). While this is in
line with ``child penalty'' literature starting from Kleven et al.
(\citeproc{ref-kleven2019a}{2019}), this specification includes the
individual fixed effect to make a unit comparison. The importance of the
individual fixed effect for the child penalty is discussed in
Arkhangelsky et al. (\citeproc{ref-arkhangelsky2026}{2026}).

Figure~\ref{fig-child-penalty} shows the coefficient \(\beta_k\) of the
event study Eq.~\ref{eq-event-study}. For women, childbirth has a huge
negative impact on working hours (-25 hours at \(q = 1\)) and a positive
impact on domestic labor (59 hours). The impact on men is relatively
small (-1 hours for working hours and 8 hours for domestic labor).
Interestingly, the impact on women's leisure is negative and larger (-34
hours at \(q = 1\)) than men's leisure (-7 hours). This implies that a
pure specialization is not held. If the difference in their wages is the
reason why husbands work more in the market and wives work more in the
house, the decline of leisure after childbirth should be similar for men
and women. This child penalty on leisure might discourage women from
having children.\footnote{ Guo and Xie (\citeproc{ref-guo2024}{2024})
  also show the child penalty on leisure in Japan by the specification
  of Kleven et al. (\citeproc{ref-kleven2019a}{2019}).}

\begin{figure}

\centering{

\includegraphics[width=0.9\linewidth,height=\textheight,keepaspectratio]{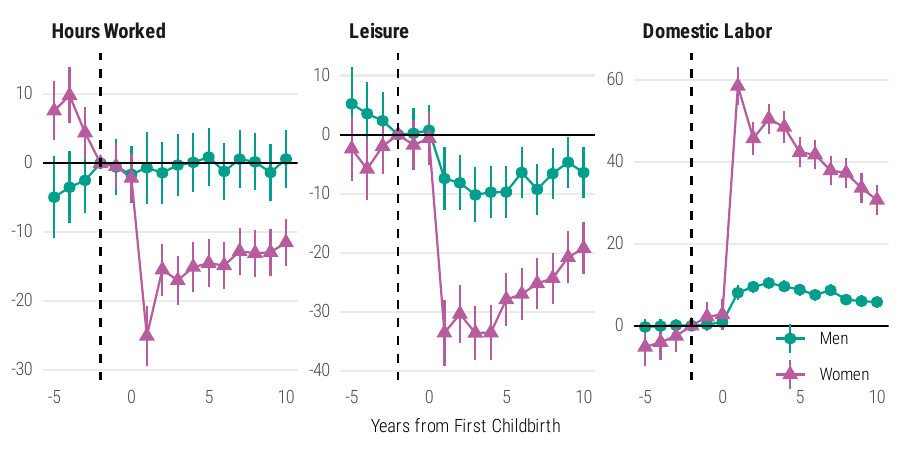}

}

\caption{\label{fig-child-penalty}\textbf{Event Study of Time Allocation
by Childbirth}. The child penalty on leisure for singles and married
couples. The data is from the JHPS, and the sample is restricted to the
age group 25-54 in 2005-2023.}

\end{figure}%

\subsection{Intra-Household Time Allocation and Bargaining
Power}\label{intra-household-time-allocation-and-bargaining-power}

Where does the child penalty on leisure come from? One possible
explanation is the intra-household bargaining. If a husband has a higher
wage than his wife, he might have more bargaining power on household
time allocations. As a result, the wife might have less leisure time
than the husband.

To highlight the role of bargaining, I study how the relative wages of
husbands and wives affect leisure time allocations within households.
Figure~\ref{fig-intra-household} shows the relationship between the log
difference in wages and the leisure time allocation by the existence of
small children (younger than 7 years old). An interesting result is that
the relationship for leisure is positive, i.e., the partner with higher
wages has more leisure time. This implies that the husband has more
bargaining power on leisure and might suggest that the bargaining power
is a mechanism behind the child penalty on leisure. The figure also
shows that the intra-household gaps in leisure time are larger for
couples with small children than for those without small children. This
implies that the wife's bargaining power become weaker when children are
born and is consistent with the child penalty on leisure.

\begin{figure}

\centering{

\includegraphics[width=0.9\linewidth,height=\textheight,keepaspectratio]{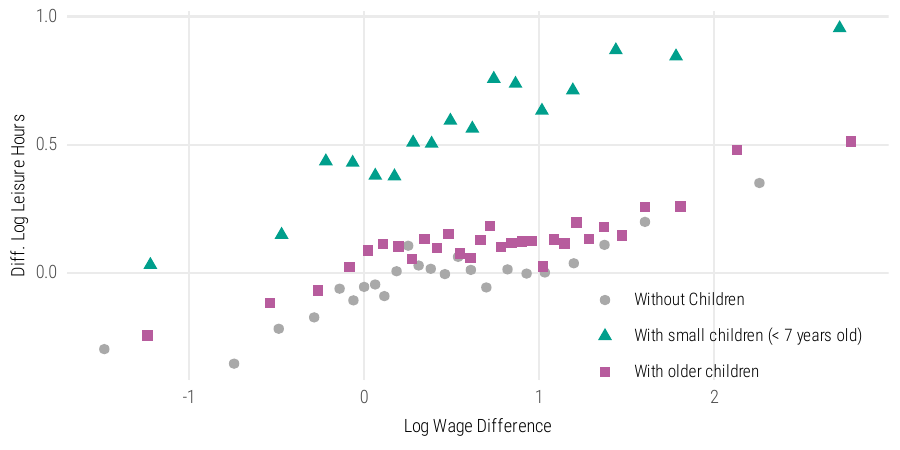}

}

\caption{\label{fig-intra-household}\textbf{Intra-Household Leisure Time
Allocation}. A bin-scatter plot of the intra-household leisure time
allocation for dual-working married couples. The data is from the JHPS,
and the sample is restricted to the age group 25-54 from 2005 to 2023.}

\end{figure}%

\section{Model}\label{sec-model}

\subsection{Setup}\label{setup}

The model economy consists of an equal mass continuum of men (gender
\(m\)) and women (gender \(f\)). Individuals are distinguished by their
discrete level of exogenous hourly wage, denoted
\(w_{g} \in \mathcal{W}\) for gender \(g \in \{m, f\}\). It will be
assumed that \(w_g\) is log-normally distributed so that
\(\log w_g \sim \mathcal{N}(\mu_{w_g}, \sigma_{w_g})\). Some individuals
of each gender will be married, but the rest will never marry. Divorce
is not allowed in this model. Married couples can have up to three
children, while singles are not allowed to have children.

Individuals live forever but stochastically age and die. There are three
states of life: young (\(Y\)), middle-aged (\(M\)), and old (\(O\)). In
each period, a person faces a constant probability \(\kappa\) of
transitioning to the next life stage: from \(Y\) to \(M\), from \(M\) to
\(O\), and from \(O\) to death. Upon death, an individual is replaced by
a new single individual with the age \(Y\) and the wage
\(w_g \in \mathcal{W}\).

Marriage and having a new child are allowed only for the couple with age
\(Y\) or \(M\), and they can only marry someone of the same age. Also,
there are two states of the age of children, \(0\) and \(1\). Children
age from \(0\) to \(1\) when their parents age from \(Y\) to \(M\) or
\(M\) to \(O\).

Each individual is endowed with one unit of time each period. Households
(single or couple) are also required to perform domestic labor based on
their marital status and the age and number of children. This
requirement captures housework and childcare that have to be done by the
household members. The household decides how to allocate time for the
labor market, leisure, and domestic labor. Domestic labor decisions
imply how time requirement is shared between the household members. A
married couple with age \(Y\) or \(M\) also decides to have a new child
or not if they do not have three children yet.

At the end of each period, a single person will meet someone else of the
opposite gender from the set of singles of the same age. The couple will
then draw a match-specific bliss shock \(b \in \mathcal{B}\), taken from
a distribution \(G(b)\), which is assumed to be a normal distribution
\(\mathcal{N}(\mu_b, \sigma_b)\). In a marriage, the bliss shock stays
constant over time, and the couple enjoys the same shock during each
period.

Last, let everyone have a common subjective discount factor \(\beta\).
Combined with the transition probability \(\kappa\), the effective
discount factor is \(\beta(1-\kappa)\) for all ages. Suppose that for
married, tastes over the consumption of market goods \(c\), leisure
hours \(l\), and the number of children \(N\) are represented by

\begin{equation}\protect\phantomsection\label{eq-function-utility}{
u(c, l, N) = \frac{c^{1-\gamma_c}}{1-\gamma_c} + \alpha_l \frac{l^{1-\gamma_l}}{1-\gamma_l} + \alpha_n \frac{(1 + N)^{1-\gamma_n} - 1}{1-\gamma_n}.
}\end{equation}

Note that the singles do not have children in this model, so the number
of children \(N\) is always zero, and their utility function is
simplified as

\[
u(c, l) = \frac{c^{1-\gamma_c}}{1-\gamma_c} + \alpha_l \frac{l^{1-\gamma_l}}{1-\gamma_l}.
\]

\subsection{Singles}\label{singles}

Consider the consumption and time allocation decision facing a single
individual. This is a purely static problem and does not depend on the
age of the single. For a single individual with wage \(w_g\), the
problem is given by

\begin{equation}\protect\phantomsection\label{eq-utility-single}{
v_g(w_g) = \max_{c_g, h_g, l_g, d_g} u(c_g, l_g),
}\end{equation}

subject to

\[
\begin{aligned}
c &= w_g h_g, & \text{(Budget Constraint)} \\
d_g &= \psi_g^S, & \text{(Domestic Labor Constraint)}
\end{aligned}
\]

and

\[
h_g + l_g + d_g = 1. \text{ (Time Constraint)}
\]

The domestic labor requirement for singles is a gender-specific
constant: \(\psi_m^S\) for men and \(\psi_f^S\) for women.

Next, consider the marriage decision facing a single. Suppose a single
individual of wage \(w_g\) with age \(a \in \{Y, M\}\) meets an opposite
gender single of wage \(w_{g'}\), and the potential couple draws a bliss
shock \(b\). They decide to marry based on the expected lifetime utility
of being single and being married. Let \(W_{g}^a(w_{g})\) be the
expected lifetime utilities. Both parties will realize if they remain
single in the current period. Likewise, let \(V_{g}(w_g, w_{g'}, b)\) be
the expected lifetime utility associated with a marriage in the current
period. A marriage will occur if and only if

\begin{equation}\protect\phantomsection\label{eq-marriage}{
V_m^a(w_m, w_f, b) > W_m^a(w_m) \text{ and } V_f^a(w_f, w_m, b) > W_f^a(w_f).
}\end{equation}

Define an indicator function \(\mathbb{1}^a(w_m, w_f, b)\) as taking the
value 1 if the couple marries and 0 otherwise. The value function for a
single individual of wage \(w_g\) with age \(a\) can be written as

\begin{equation}\protect\phantomsection\label{eq-value-single}{
\begin{aligned}
W_{g}^Y(w_{g}) &= v(w_{g}) \\
&+ \beta (1 - \kappa) \int_{\mathcal{B}}\int_{\mathcal{W}}
(1 - \mathbb{1}^Y) W_{g}^Y(w_g) + \mathbb{1}^Y V_{g}^Y(w_{g}, w_{g'}, 0, 0; b) \,d\hat{S}^Y_{g'}(w_{g'}) dG(b) \\
&+ \beta \kappa \int_{\mathcal{B}}\int_{\mathcal{W}} (1 - \mathbb{1}^M) W_{g}^M(w_{g}) + \mathbb{1}^M V_{g}(w_{g}, w_{g'}, 0, 0; b) \,d\hat{S}^M_{g'}(w_{g'}) dG(b) \\
W_{g}^M(w_{g}) &= v(w_{g}) \\
&+ \beta (1 - \kappa) \int_{\mathcal{B}}\int_{\mathcal{W}}
(1 - \mathbb{1}^M) W_{g}^M(w_{g}) + \mathbb{1}^M V_{g}^M(w_{g}, w_{g'}, 0, 0; b) \,d\hat{S}^M_{g'}(w_{g'}) dG(b) \\
&+ \beta \kappa W_{g}^O(w_{g}) \\
W_{g}^O(w_{g}) &= v(w_{g}) + \beta (1-\kappa) W_{g}^O(w_{g})
\end{aligned}
}\end{equation}

A single individual at age \(Y\) enjoys \(v(w_g)\) in this period. Next
period, he transitions to \(M\) with probability \(\kappa\) or remains
at age \(Y\) with probability \(1 - \kappa\). Then they meet another
opposite gender single from the distribution
\(\hat{S}_{g'}^{a}(w_{g'})\) for \(a \in \{Y, M\}\). Since marriage
behavior is different by wage, the wage distribution of the potential
partner is an endogenous object, which will be determined in the
equilibrium and formally defined in Section~\ref{sec-mdl-equilibrium}.

\subsection{Couples}\label{couples}

The consumption and time allocation decisions a couple faces are also
static problems. A couple with age \(a \in \{Y, M, O\}\), wage \(w_m\)
and \(w_f\), and \(N_0\) children of age \(0\) and \(N_1\) children of
age \(1\) solves

\begin{equation}\protect\phantomsection\label{eq-utility-married}{
\max_{c, h_m, h_f, l_m, l_f, d_m, d_f}
(1 - \lambda)u\left(\frac{c}{\Gamma(N_0, N_1)}, l_m, N_0+N_1\right) +
\lambda u\left(\frac{c}{\Gamma(N_0, N_1)}, l_f, N_0+N_1\right)
}\end{equation}

subject to

\[
\begin{aligned}
c &= w_m h_m + w_f h_f, & \text{(Budget Constraint)} \\
D(d_m, d_f) &= \psi_0 + \psi_1 \mathbb{1}\{N_0 > 0\} + \psi_2 \mathbb{1}\{N_0 + N_1 > 0\}, & \text{ (Domestic Labor Constraint)} \\
\end{aligned}
\]

and

\[
\begin{aligned}
h_g + l_g + d_g &= 1 \text{ for } g \in \{m, f\} & \text{(Time Constraint)}. \\
\end{aligned}
\]

The bargaining power of the wife is given by

\begin{equation}\protect\phantomsection\label{eq-bargaining-power}{
\lambda = \Lambda(w_m, w_f, N_0) = \frac{1}{1 + \exp\left(\rho_0 + \rho_1 (\log w_m - \log w_f) + \rho_2 \mathbb{1} \{N_0 > 0\} \right)}.
}\end{equation}

The curvature parameters \(\rho_0, \rho_1, \rho_2\) captures the
strength of the bargaining power by relative wage and existence of the
small children. If \(\rho_0 = \rho_1 = \rho_2 = 0\), the bargaining
power becomes equal (\(\Lambda(w_m, w_f, N_0) = \frac{1}{2}\)) and does
not depend on wage difference. If \(\rho_1 = 1\) and
\(\rho_0 = \rho_2\), the bargaining power is determined by the share of
the wife's wage in the total wage (\emph{a la} Baudin et al.
(\citeproc{ref-baudin2015}{2015}) ). Section~\ref{sec-app-bargaining}
shows the positive correlation between the intra-household wage gaps and
leisure time happens if and only if the bargaining power \(\rho_1\) is
larger than 1. It is also shown that the intra-household leisure time
gap by the existence of the small children \((N_0 > 0)\) if
\(\rho_2 > 0\).

The \(\psi_0\) represents the domestic labor requirement for married
couples without children, and \(\psi_2\) represents the additional
domestic labor requirement. If they have a small child, the domestic
labor requirement increases by \(\psi_1\).\footnote{ In the Appendix,
  Figure~\ref{fig-hours-domestic} shows that domestic labor hours
  significantly increase with the existence of children but do not
  change much with the number of children.} The \(\Gamma(N_0, N_1)\) is
the consumption scale factor that depends on the number of children.
Here I assume the economies of scale, i.e.,
\(\Gamma(N_0, N_1) < 2 + N_0 + N_1\).

The production function of domestic labor is given by a Constant
Elasticity of Substitution (CES) function,

\begin{equation}\protect\phantomsection\label{eq-joint-domestic-labor}{
D(d_m, d_f) = \left((1-\theta) d_m^{\xi} + \theta d_f^{\xi} \right)^{\frac{1}{\xi}}
\quad \text{with } \theta \in (0, 1), \xi < 1.
}\end{equation}

The \(\theta\) parameter represents the relative productivity of the
wife in domestic labor and potentially captures the social norms of
gender roles in domestic labor. This is because the higher \(\theta\)
incentivizes the wife to do more domestic labor and the husband to do
less. I formulate the role of \(\theta\) in Section~\ref{sec-app-theta}.
The \(\xi\) parameter represents the elasticity of substitution between
the husband's and wife's domestic labor. If \(0 < \xi < 1\), the
domestic labor is a substitute, and if \(\xi < 0\), the domestic labor
is a complement.

For ages \(Y\) or \(M\), the couple also decides whether to have a new
child or not. I assume the couple can have up to three children. The
couple decides to have a new child with the same bargaining power
\(\lambda = \Lambda(w_m, w_f, N_0)\) as the time allocation decision.
Hence, the young couple with \(N_0 \le 2\) solves

\begin{equation}\protect\phantomsection\label{eq-childbirth-young}{
\begin{aligned}
\max_{N'_0 \in \{N_0, N_0 + 1\}}
& (1-\kappa) \left[
  (1 - \lambda) V_m^Y(w_m, w_f, N'_0, 0; b) +
  \lambda V_f^Y(w_f, w_m, N'_0, 0; b) \right]\\
&+ \kappa \left[
  (1 - \lambda) V_m^M(w_m, w_f, 0, N'_0; b) +
  \lambda V_f^M(w_f, w_m, 0, N'_0; b) \right]
\end{aligned}
}\end{equation}

and the middle-aged couple with \(N_0 + N_1 \le 2\) solves

\begin{equation}\protect\phantomsection\label{eq-childbirth-middle}{
\begin{aligned}
\max_{N'_0 \in \{N_0, N_0 + 1\}}
& (1-\kappa) \left[
  (1 - \lambda) V_m^M(w_m, w_f, N'_0, N_1; b) +
  \lambda V_f^M(w_f, w_m, N'_0, N_1; b) \right]\\
&+ \kappa \left[
  (1 - \lambda) V_m^O(w_m, w_f, 0, N'_0 + N_1; b) +
  \lambda V_f^O(w_f, w_m, 0, N'_0 + N_1; b) \right].
\end{aligned}
}\end{equation}

The first term in Eq.~\ref{eq-childbirth-young} captures the couple's
expected lifetime utility if they do not age to \(M\). If they don't
have three children yet, they can choose to have a new child or not, and
the number of children at age \(0\) will be \(N'_0\). The second term
represents the expected lifetime utility of the couple if they age to
\(M\). When they age to \(M\), their children will age from \(0\) to
\(1\), and as a result, the number of children at age \(0\) will be
zero, and the number of children at age \(1\) will be \(N'_0\). Note
that this stochastic aging of children also applies to their newborn
child, i.e., the newborn child will be aged to \(1\) in the next
period.\footnote{ This assumption is made for computational reasons, to
  reduce the dimensions of the state space.} The middle-aged couple's
decision Eq.~\ref{eq-childbirth-middle} is similar to the young couple's
decision.

Given married couple's decision of childbirth, the newborn child will be
born with a probability \(\delta_1\) when the parents aged \(Y\) and
with a probability \(\delta_2\) when the parents aged \(M\). These
probabilities capture heterogeniety and uncertainty in the childbirth
decision as well as the difference in fecundity by age. The probability
of not having a child even if they want could be increased with age.

Let the indirect utility functions derived from the couple's time
allocation problem be \(v_g(w_g, w_{g'}, N_0, N_1)\) for
\(g \in \{m, f\}\), which does not depend on the age of the couple.
Given the childbirth decision \(N_0^*\), the value function for a
married man of wage \(w_m\) and with \(N_0\) children age \(0\) and
\(N_1\) children age \(1\) and a bliss shock \(b\) can be written as

\begin{equation}\protect\phantomsection\label{eq-value-married}{
\begin{aligned}
&V_{g}^Y(w_{g}, w_{g'}, N_0, 0; b) = v_{g}(w_{g}, w_{g'}, N_0, 0) + b \\
&\quad+ \beta (1 - \kappa) \delta_1 V_{g}^Y(w_{g}, w_{g'}, N_0^*, 0; b)
+ \beta (1 - \kappa) (1-\delta_1)V_{g}^Y(w_{g}, w_{g'}, N_0, 0; b) \\
&\quad+ \beta \kappa \delta_1 V_{g}^M(w_{g}, w_{g'}, 0, N_0^*; b)
+ \beta \kappa (1-\delta_1)V_{g}^M(w_{g}, w_{g'}, 0, N_0; b) \\
&V_{g}^M(w_{g}, w_{g'}, N_0, N_1; b) = v_{g}(w_{g}, w_{g'}, N_0, N_1) + b \\
&\quad+ \beta (1 - \kappa) \delta_2 V_{g}^M(w_{g}, w_{g'}, N_0^*, N_1; b)
+ \beta (1 - \kappa) (1-\delta_2)V_{g}^M(w_{g}, w_{g'}, N_0, N_1; b) \\
&\quad+ \beta \kappa \delta_2 V_{g}^O(w_{g}, w_{g'}, 0, N_0^* + N_1; b)
+ \beta \kappa (1-\delta_2) V_{g}^O(w_{g}, w_{g'}, 0, N_0 + N_1; b) \\
&V_{g}^O(w_{g}, w_{g'}, 0, N_1; b) = v_{g}(w_{g}, w_{g'}, 0, N_1) + b
+ \beta (1-\kappa) V_{g}^O(w_{g}, w_{g'}, 0, N_1; b).
\end{aligned}
}\end{equation}

\subsection{Equilibrium}\label{sec-mdl-equilibrium}

The dynamic programming problem for a single person, or equation
Eq.~\ref{eq-value-single}, depends on the problem's solution for a
married person, as given by equation Eq.~\ref{eq-value-married}. In
addition, solving the single's problem requires knowing the steady-state
wage distribution of potential mates (opposite gender \(g'\)) in the
marriage market \(S_{g'}^{a}\) for \(a \in \{Y, M, O\}\). The
non-normalized steady-state wage distributions for singles are given by

\begin{equation}\protect\phantomsection\label{eq-dist-single}{
\begin{aligned}
S_{g}^Y(w_{g}) &= (1 - \kappa)\int_{\mathcal{B}}\int_{\mathcal{W}_{g'}}
\int_{\mathcal{W}_{g}}^{w_{g}}
(1 - \mathbb{1}^Y(w'_g, w'_{g'}, b))\,dS_g^Y(w'_g)d\hat{S}_{g'}^Y(w'_{g'})dG(b) \\
&+ \frac{\kappa}{3} \int_{\mathcal{W}}^{w_g}\,dF_{g}(w'_g), \\
S_{g}^M(w_{g}) &= \kappa\int_{\mathcal{B}}\int_{\mathcal{W}} \int_{\mathcal{W}}^{w_{g}}
(1 - \mathbb{1}^Y(w'_{g}, w'_{g'}, b))\,dS_{g}^Y(w'_{g})d\hat{S}_{g'}^Y(w'_{g'})dG(b) \\
&+ (1 - \kappa) \int_{\mathcal{B}}\int_{\mathcal{W}} \int_{\mathcal{W}}^{w_{g}}
(1 - \mathbb{1}^M(w'_{g}, w'_{g'}, b))\,dS_{g}^M(w'_{g})d\hat{S}_{g}^M(w'_{g'})dG(b), \\
S_{g}^O(w_{g}) &= \kappa\int_{\mathcal{B}}\int_{\mathcal{W}} \int_{\mathcal{W}}^{w_{g}}
(1 - \mathbb{1}^M(w'_{g}, w'_{g'}, b))\,dS_{g}^M(w'_{g})d\hat{S}_{g'}^M(w'_{g'})dG(b) \\
&+ (1 - \kappa) \int_{\mathcal{W}}^{w_{g}}\,dS_{g}^O(w'_{g}).
\end{aligned}
}\end{equation}

In the above equations, \(\hat{S}_{g'}^{a}\) represents the normalized
distribution of singles of the opposite gender for age
\(a \in \{Y, M\}\) and is given by

\begin{equation}\protect\phantomsection\label{eq-dist-single-normalized}{
\hat{S}_{g'}^{a}(w_f) = \frac{S_{g'}^{a}(w_{g'})}{\int_{\mathcal{W}} \,dS_{g'}^{a}(w_f')}.
}\end{equation}

The first term of the \(S_g^Y(w_g)\) in equation
Eq.~\ref{eq-dist-single} counts those singles who did not marry in the
current period and did not age to \(M\). The second term represents the
arrival of new adults. The mass of new arrivals is normalized as the
total mass of singles and married couples is 1 in the steady state. For
its derivation, see Section~\ref{sec-app-mass}. \(S_g^M(w_g)\) consists
of the flow of singles who did not marry in the current period and aged
to \(M\) and singles at age \(M\) who remained single and did not age to
\(O\). Since singles at age \(O\) do not marry, \(S_g^O(w_g)\) consists
of singles at age \(M\) who did not marry in the current period and aged
to \(O\) and singles at age \(O\) who did not die in the current period.

\begin{definition}[Stationary Matching
Equilibrium]\protect\hypertarget{def-steady-state}{}\label{def-steady-state}

A \emph{stationary matching equilibrium} is a set of value functions for
singles \(W_m^a(w_m)\) and \(W_f^a(w_f)\) and couples
\(V_m^a(w_m, w_f, N_0, N_1; b)\) and \(V_f^a(w_f, w_m, N_0, N_1; b)\);
matching rules for singles \(\mathbb{1}^a(w_m, w_f, b)\); and
\emph{stationary distributions for singles} \(S_m^a(w_m)\) and
\(S_f^a(w_f)\) such that:

\begin{enumerate}
\def\labelenumi{\arabic{enumi}.}
\tightlist
\item
  The value function \(W_m^a(w_m)\) and \(W_f^a(w_f)\) solve the
  single's recursion Eq.~\ref{eq-value-single}, taking as given their
  indirect utility functions \(v_m(w_m)\) and \(v_f(w_f)\) from problem
  Eq.~\ref{eq-utility-single}, the value functions for couples
  \(V_m^a(w_m, w_f, N_0, N_1; b)\) and \(V_f^a(w_f, w_m, N_0, N_1; b)\),
  the matching rule for singles \(\mathbb{1}^a(w_m, w_f, b)\) from
  Eq.~\ref{eq-marriage}, and the wage distribution of potential mates
  \(\hat{S}_{m}^a(w_{m})\) and \(\hat{S}_{f}^a(w_{f})\) defined in
  Eq.~\ref{eq-dist-single-normalized}.
\item
  The value function \(V_{m}^a(w_m, w_f, N_0, N_1; b)\) and
  \(V_{f}^a(w_f, w_m, N_0, N_1; b)\) solve the couple's recursion
  Eq.~\ref{eq-value-married}, taking as given their indirect utility
  functions \(v_{m}(w_m, w_f, N_0, N_1)\) and
  \(v_{f}(w_f, w_m, N_0, N_1)\) from the couple's problem
  Eq.~\ref{eq-utility-married}, and the childbirth decision \(N_0^*\)
  from the couple's problem Eq.~\ref{eq-childbirth-young} or
  Eq.~\ref{eq-childbirth-middle}.
\item
  The matching rule for singles \(\mathbb{1}^a(w_m, w_f, b)\) is
  determined by the equation Eq.~\ref{eq-marriage}, taking as given the
  value functions for \(W_m^a(w_m)\), \(W_f^a(w_f)\),
  \(V_m^a(w_m, w_f, N_0, N_1; b)\), and
  \(V_f^a(w_f, w_m, N_0, N_1; b)\).
\item
  The stationary distribution for singles \(S_g^a(w_g)\) solves the
  equation Eq.~\ref{eq-dist-single}, taking as given the matching rule
  for singles \(\mathbb{1}^a(w_m, w_f, b)\).
\end{enumerate}

\end{definition}

\section{Calibration}\label{sec-calib}

The model developed will now be fit to the Japanese data for the period
2019-2023. Some parameters are exogenously determined based on a priori
information or taken directly from data. Most of the parameters,
however, will be estimated using a minimum distance procedure. In
Section~\ref{sec-past}, the model will be simulated using female wages,
social norms, and leisure technology from the period 2005-2009. It will
be assumed that the model is in a steady state for each of these years.
A comparison between two steady states will determine the key factors
that can account for changes in marriage and fertility behavior.

\subsection{Exogenous Parameters}\label{exogenous-parameters}

The length of a model period is one year. Let \(\beta\) (the subjective
discount factor) be 0.96, as the standard value in macroeconomics
studies, such as in Prescott (\citeproc{ref-prescott1986}{1986}). All
the targets for the estimation are calculated for individuals between
the ages of 25 and 54, which corresponds to an operational lifespan of
30 years. Let the transition probability \(\kappa = 1/10 = 0.1\), so
that individuals in the model also live 30 years on average. Finally,
following the Organization for Economic Co-operation and Development
(OECD) equivalence scale, set \(\chi_0 = 0.5\) and \(\chi_1 = 0.3\).

Next, some parameters are directly computed from the data. Given the
time constraint of the single's problem Eq.~\ref{eq-utility-single},
domestic labor requirements for singles, \(\psi_m^S\) and \(\psi_f^S\),
are equal to their time spent on domestic labor, \(d_m\) and \(d_f\).
Using the mean value of singles' \(d_m\) and \(d_f\) from the JHPS in
the period 2019-2023, set \(\psi_m^S =\) 0.032 and \(\psi_f^S =\) 0.056.
The mean wage of men is normalized to 1, so \(\mu_{w_m} = 0\). In
addition, I use the standard deviation of the log wage of men in the
JHPS in the period 2019-2023 as the \(\sigma_{w_m}\) since almost all
men in the sample work in the market (all the singles by the sample
definition and 95.7\% of married men in the data). The standard
deviation of the log wage of women will be endogenously determined in
the estimation since some of married women do not work in the market
(20.1\% in the data).

To sum up, the parameter values exogenously determined are summarized in
Table~\ref{tbl-param-exog}.

\begin{table}

\caption{\label{tbl-param-exog}Parameters (A Priori Information)}

\centering{

\centering
\begin{tblr}[         %% tabularray outer open
]                     %% tabularray outer close
{                     %% tabularray inner open
colspec={Q[]Q[]},
hline{2}={1-2}{solid, black, 0.05em},
hline{1}={1-2}{solid, black, 0.1em},
hline{8}={1-2}{solid, black, 0.1em},
}                     %% tabularray inner close
Parameter & Source \\
$\Gamma (N) = 1 + 0.5 + 0.3N$ & OECD equivalence scale \\
$\beta = 0.96$ & Prescott (1986) \\
$\kappa = 1 / 10$ & 30-year lifespan \\
$\mu_{w_m} = 0$ & Male wage is normalized to 1 \\
$\psi_m^{S}=0.032$, $\psi_f^{S}=0.056$ & JHPS2019-2023 \\
$\sigma_{w_m}=0.706$ & JHPS2019-2023 \\
\end{tblr}

}

\end{table}%

\subsection{Endogenous Parameters}\label{endogenous-parameters}

Now, 19 parameters will be estimated using a minimum distance procedure.
There are five preference parameters,
\(\{\gamma_c, \gamma_l, \gamma_n, \alpha_l, \alpha_n\}\); three
bargaining power parameter \(\{\rho_0, \rho_1, \rho_2\}\); two female
wage distribution parameters, \(\{\mu_{w_f}, \sigma_{w_f}\}\); two
parameters for bliss shock, \(\{\mu_{b}, \sigma_{b}\}\); two home
production parameters \(\{\theta, \xi\}\) ; three domestic labor
requirements \(\{\psi_0, \psi_1, \psi_2\}\); and two fertility
parameters \(\{\delta_1, \delta_2\}\).

The data targets are as follows:

\begin{itemize}
\tightlist
\item
  \emph{Leisure time for singles}: The mean of leisure time for singles
  aged 25-54 in the JHPS2019-2023.
\item
  \emph{Leisure and domestic labor time for couples}: The mean of
  leisure time for married couples aged 25-54 in the JHPS2019-2023. For
  each status of children, without children, with small children
  (\textless{} 7 years old), and with older children (7-18 years old),
  the mean of leisure time and domestic labor time are targeted.
\item
  \emph{Marriage rate}: The share of never-married women at age 45-54 in
  the 2020 Japanese Census. The model moment is also computed for the
  women at age \(O\).
\item
  \emph{Number of children}: The share of women at 44 years old with one
  child, two children, and three or more children in the Human Fertility
  Database (HFD) 2019-2023, i.e., the 1975-1979 cohorts. Model moments
  are computed for women at age \(O\).
\item
  \emph{Labor market outcomes}: The gender gaps in single's log wage and
  the standard deviation of the log wage for single women in the
  JHPS2019-2023.
\end{itemize}

As Table~\ref{tbl-smm-fit} illustrates, the model has no problem
matching most of the targets.

\section{Baseline Economy}\label{sec-baseline}

\subsection{Estimated Parameters}\label{estimated-parameters}

Table~\ref{tbl-endo-params} shows the estimated parameter values. The
estimate of the degree of curvature in the utility function for market
goods (\(\gamma_c =\) 1.582) is in line with the macroeconomics
literature, which typically uses a coefficient of relative risk aversion
of either 1 or 2. The other curvatures, the one for leisure time
(\(\gamma_l =\) 1.341) and for number of children (\(\gamma_n =\)
1.300), are also in the range between 1 and 2. The preference strength
of leisure (\(\alpha_l =\) 2.335) is significantly higher than the one
for market goods (since it is normalized to 1), which may reflect the
high value of leisure time these days.

The bargaining power parameter \(\rho_0 =\) -0.267 implies that the
wife's bargaining power is higher than the husband's when their wages
are equal and they do not have small children. It reduces the single
men's motivation to get married, and this is consistent with the fact
that one of the reasons for not getting married is the fear of losing
freedom and leisure time (Figure~\ref{fig-reason-single}). The other
values \(\rho_1 =\) 1.463 and \(\rho_2 =\) 0.782 suggest that the wife's
bargaining power is lower when the husband's wage is higher than the
wife's wage and when they have small children. As shown in
Section~\ref{sec-app-bargaining}, \(\rho_1 > 1\) suggests that the
intra-household leisure time gap is positively correlated with the wage
and \(\rho_2 >0\) suggests the gap is larger when they have small
children. The estimated values are consistent with the observed
relationship between the intra-household wage gap and leisure time in
Figure~\ref{fig-intra-household}.

While the variance of the log wage (\(\sigma_{w_f} =\) 0.771) takes
close values to the observed log wage dispersion of the single women
(0.766), the gender gap in the mean of the log wage
(\(\mu_{w_m}-\mu_{w_f} =\) 0.139) is larger than the gaps in singles
data (0.148). It suggests that the marriage rate gaps between people
with high and low wages, and it will be shown in
Figure~\ref{fig-marriage-earn} in Section~\ref{sec-marriage-earn}.

The bliss shock parameters, \(\mu_b =\) -1.603 and \(\sigma_b =\) 1.326,
suggest that the bliss shock is negative in the most of the matching
(88.7\%). This is a down force for the singles to marry while marriage
could improve their utility by the economy of scale (\(\chi_0 = 0.5\))
and the possibility of having children.

The CES function parameters, \(\theta =\) 0.831 and \(\xi =\) -0.029,
suggest that the domestic labor is a weak substitute (\(\xi > 0\)) and
the wife's productivity in domestic labor is higher than the husband's
(\(\theta > 0.5\)). This can be interpreted as social norms that make
wives work more at home. The domestic labor requirements, \(\psi_0 =\)
0.110, \(\psi_1 =\) 0.226, and \(\psi_2 =\) 0.052, are reasonable values
since the small children requires more domestic labor than the older
children.

Finally, the childbirth probabilities, \(\delta_1 =\) 0.254 and
\(\delta_2 =\) 0.193, suggest that the probability of having a child is
higher when the parents are young and lower when they are middle-aged.
It incentivizes the singles to marry earlier if they want to have
children.

\begin{table}

\caption{\label{tbl-endo-params}Parameters Estimated (Minimum Distance)}

\centering{

\centering
\begin{tblr}[         %% tabularray outer open
]                     %% tabularray outer close
{                     %% tabularray inner open
width={0.8\linewidth},
colspec={X[]X[]},
hline{2}={1-2}{solid, black, 0.05em},
hline{1}={1-2}{solid, black, 0.1em},
hline{10}={1-2}{solid, black, 0.1em},
}                     %% tabularray inner close
Category & Parameter Values \\
Preference & $\gamma_c = 1.582, \gamma_l = 1.341, \gamma_n = 1.300$ \\
& $\alpha_l = 2.335, \alpha_n = 3.211$ \\
Bargaining & $\rho_0 = -0.267, \rho_1 = 1.463, \rho_2 = 0.782$ \\
Female wage & $\mu_{w_f} = -0.139, \sigma_{w_f} = 0.771$ \\
Match quality & $\mu_b = -1.603, \sigma_b = 1.326$ \\
Home production & $\theta = 0.831, \xi = -0.029$ \\
Domestic labor & $\psi_0 = 0.110, \psi_1 = 0.226, \psi_2 = 0.052$ \\
Fertility & $\delta_1 = 0.254, \delta_2 = 0.193$ \\
\end{tblr}

}

\end{table}%

\subsection{Marriage Rate by Earnings}\label{sec-marriage-earn}

Next, I show how the model economy performs along dimensions that are
not targeted in the calibration. Figure~\ref{fig-marriage-earn} shows
the marriage rate by earnings deciles, i.e., the first decile contains
0-10\% of the population by earnings, and so on. The bars show the share
of married individuals with age \(O\) in each decile of earnings in the
baseline model. For the data moments, I use the Employment Status Survey
2022, which reports the number of people by earnings, age, and marital
status. Since the data reports categorical earnings, I compute the
cumulative density function of the earnings distribution and calibrate
the marriage rate by earnings deciles by linear interpolation. For the
details, see Section~\ref{sec-app-ess}.

\begin{figure}

\centering{

\includegraphics[width=0.9\linewidth,height=\textheight,keepaspectratio]{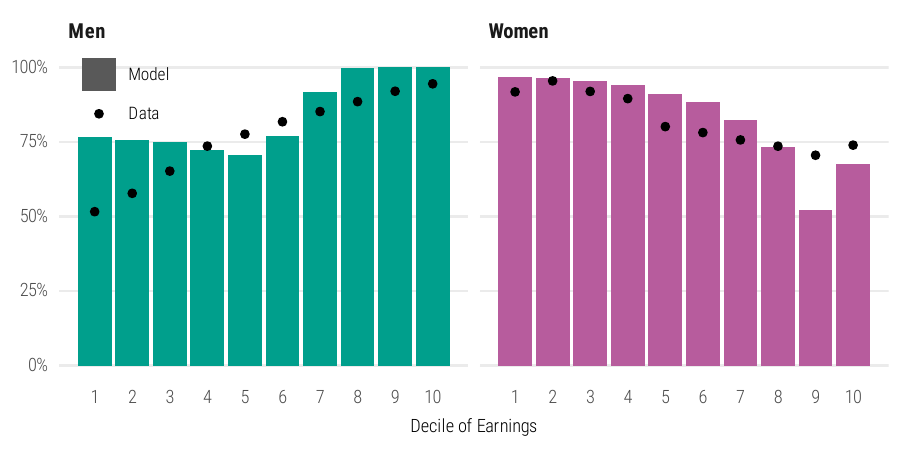}

}

\caption{\label{fig-marriage-earn}\textbf{Marriage Rate by Earnings
Decile}. The bars show the share of married individuals in each decile
of earnings in the baseline model. The points show the marriage rate
estimated from the Employment Status Survey 2022.}

\end{figure}%

The model economy captures the general pattern of the marriage rate by
earnings. In particular, the model economy reproduces the positive
correlation between the marriage rate and earnings for men and negative
correlation for women.

\subsection{Child Penalty on Leisure}\label{child-penalty-on-leisure-1}

In Section~\ref{sec-facts}, I show that the time allocation changes
around the first childbirth differ between men and women. Female leisure
time decreases significantly more than male leisure time.
Figure~\ref{fig-child-penalty-mdl} shows a similar event study using the
baseline model. I created a ten thousand single young men and women with
wages drawn from \(\mathcal{N}(\mu_{w_m}, \sigma_{w_m}^2)\) and
\(\mathcal{N}(\mu_{w_f}, \sigma_{w_f}^2)\), respectively. I simulate
their time allocations and life events, such as marriage, childbirth,
and death until 30 periods. The specification is the same as
Eq.~\ref{eq-event-study}.

Figure~\ref{fig-child-penalty-mdl} shows the results. To make a
comparison with the data, the hours are re-scaled to weekly hours, i.e.,
\(h + l + d = 16 \times 7 = 112\) hours. As in the data, the model
economy shows a decline in working hours and a larger increase in
domestic labor for women. The model economy also shows a decrease in
leisure time for both men and women, which is consistent with the data.
While it is also consistent with data that the decrease in leisure time
is larger for women, the model economy shows a smaller difference in the
decrease in leisure time between men and women (in one year after the
first childbirth, 26.1 hours in the data and 17.1 hours in the model).

\begin{figure}

\centering{

\includegraphics[width=0.9\linewidth,height=\textheight,keepaspectratio]{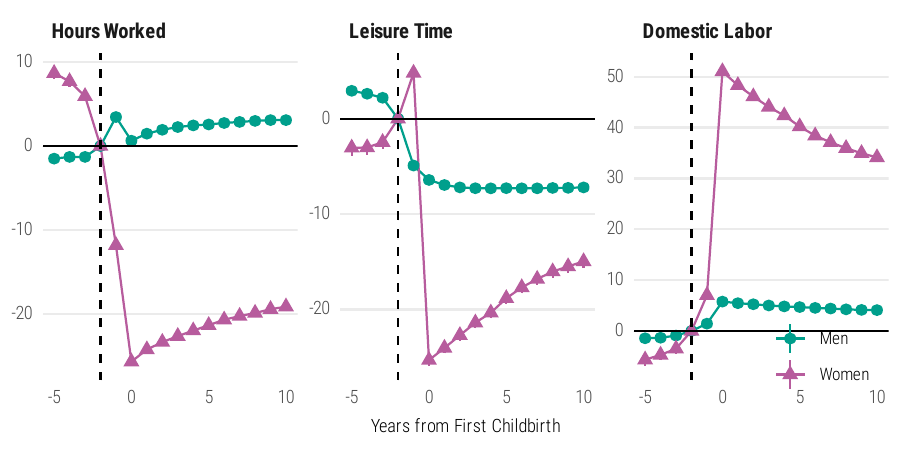}

}

\caption{\label{fig-child-penalty-mdl}\textbf{Child Penalty in the
Baseline Model}. The figure shows the event study of time allocation
around the first childbirth in the baseline model.}

\end{figure}%

\section{Back to 2005-2009}\label{sec-past}

In this section, I simulate for the 2005-2009 period. The main goal of
this analysis is to disentangle the main drivers of the recent decline
in the marriage rate and fertility rate in Japan. To answer the
question, I show the three stylized facts that significantly changed in
the past decades and may affect marriage and fertility decisions. Given
the facts, I choose three model parameters that capture the three
changes: female wage growth, change in social norms, and leisure
technology growth. With all other parameters kept at the baseline
values, I simulated the model to evaluate the impact of each factor on
the marriage and fertility decisions.

\subsection{What has changed since 2005-2009?}\label{sec-past-facts}

In this section, I will show the time trend of factors that may affect
marriage and fertility decisions. As in the previous section, I will use
the samples aged between 25-54 in the JHPS. The single sample is defined
as those who work, are not married, and have no children. The married
couple sample is defined as those who are married and have at most 3
children.

Panel (a) in Figure~\ref{fig-past-wage} shows the gender gap in the log
wage of singles. Each point represents the differences in the average
log wage and the linear fit is also plotted. From 2005, the female
relative wage has increased by around 5\%. Since this is not a small
change, I will calibrate the parameter \(\mu_{w_f}\) to capture this
trend.

Panel (b) in Figure~\ref{fig-past-wage} shows the husband's share in
domestic labor of married couples. The husbands' share in domestic labor
has almost doubled from 2005 to 2023. I interpret this as the change in
social norms on gender roles, and I will discuss its impact on marriage
and fertility decisions by calibrating the parameter \(\theta\).

\begin{figure}

\centering{

\includegraphics[width=0.9\linewidth,height=\textheight,keepaspectratio]{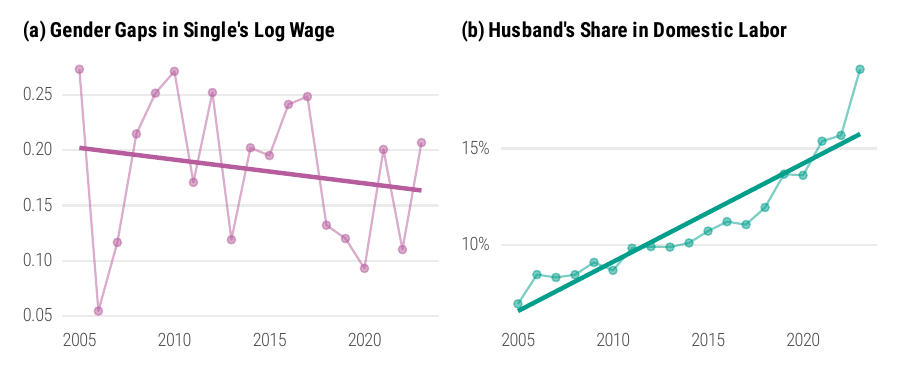}

}

\caption{\label{fig-past-wage}\textbf{Female Wage Growth and Husband's
Domestic Labor}. Panel (a) shows the gaps in the average log wages of
singles aged between 25-54. Panel (b) shows the husband's share in
domestic labor of married couples aged between 25-54. In both panels,
each dot represents the data point and the center line is the linear
fit. The data from JHPS and see the text for the sample selection.}

\end{figure}%

Finally, Figure~\ref{fig-past-hours} shows the average hours worked and
leisure time of singles. Each point represents the average hours worked
and leisure time for each year and the linear fit is also shown. The
hours worked are decreasing, and the leisure time is increasing over
time. I interpret this as the increase in leisure technology, and I will
estimate the parameter \(\alpha_l\) to study its impact on marriage and
fertility decisions.

\begin{figure}

\centering{

\includegraphics[width=0.9\linewidth,height=\textheight,keepaspectratio]{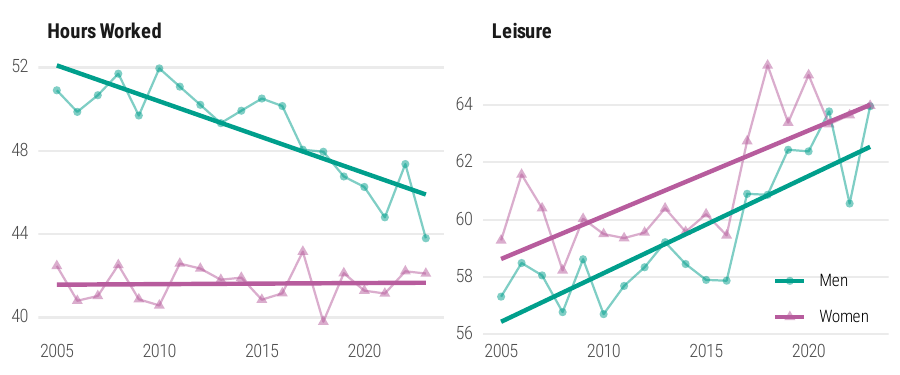}

}

\caption{\label{fig-past-hours}\textbf{Single's Hours Worked and Leisure
Time}. The figure shows the average hours worked and leisure time of
singles aged between 25-54. Each dot represents the average hours worked
and leisure time in a year and the center line is the linear fit. The
data from JHPS and see the text for the sample selection.}

\end{figure}%

\subsection{Driving Forces}\label{sec-past-params}

Section~\ref{sec-past-facts} points out the three important factors that
may affect marriage and fertility decisions: female wage growth, change
in social norms, and leisure technology growth. I interpret these
factors as the changes in the parameters \(\mu_{w_f}\), \(\theta\), and
\(\alpha_l\), respectively. The mean of log female wage, \(\mu_{w_f}\),
is directly connected to the female relative wage. The parameter
\(\theta\) is related to the social norms on domestic labor because it
captures the relative productivity of the husband for domestic labor.
The leisure technology growth is captured by the parameter \(\alpha_l\)
because larger \(\alpha_l\) means the higher utility given the same
amount of leisure.

To estimate these parameters, I use a minimum distance procedure by
fixing other parameters at the baseline values. Here I assume that all
the changes from 2005-2009 to 2019-2023 are due to these three factors,
and I choose them as the targets for the estimation.

\begin{table}

\caption{\label{tbl-past-params}Calibrated Parameters in 2005-2009 and
2019-2023}

\centering{

\centering
\begin{talltblr}[         %% tabularray outer open
entry=none,label=none,
note{}={Notes: The second column represents the estimated values of the parameters in the 2005-2009 period by minimizing the distance between data and model moments. The third column shows the targets for the estimation, and the fourth and fifth columns show the data moments and the model moments.},
]                     %% tabularray outer close
{                     %% tabularray inner open
width={0.95\linewidth},
colspec={X[0.05]X[0.15]X[0.55]X[0.1]X[0.1]},
hline{2}={1-5}{solid, black, 0.05em},
hline{1}={1-5}{solid, black, 0.1em},
hline{5}={1-5}{solid, black, 0.1em},
}                     %% tabularray inner close
& 2005-2009 & Target & Data & Model \\
$\alpha_l$ & 1.858 & Single men's leisure hours, $l_m$ & 0.516 & 0.516 \\
$\mu_f$ & -0.144 & Gender difference in single's, $\log w_m$ and $\log w_f$ & 0.181 & 0.181 \\
$\theta$ & 0.923 & Husband's share of domestic labor, $d_m / \left(d_m + d_f\right)$ & 0.917 & 0.917 \\
\end{talltblr}

}

\end{table}%

The estimated parameters are summarized in Table~\ref{tbl-past-params}.
The second column shows the estimated values of the parameters in the
2005-2009 period. The third column shows the targets for the estimation,
and the fourth and fifth columns show the data moments and the model
moments, respectively. The estimated value \(\alpha_l\) is 1.858, which
is smaller than the baseline value of 2.335. It can be interpreted as a
leisure technology growth in this period. The change from \(\mu_f =\)
-0.144 in 2005-2009 to -0.139 in 2019-2023 captures the female wage
growth. Also, the change from \(\theta =\) 0.923 in 2005-2009 to 0.831
in 2019-2023 can be interpreted as the shift in social norms.

\subsection{Why did the Marriage and Fertility
Decline?}\label{why-did-the-marriage-and-fertility-decline}

Given the calibrated parameters in Section~\ref{sec-past-params}, I
simulate the model in the period 2005-2009, and
Table~\ref{tbl-past-decomp} shows the results. The first row shows the
baseline model for the period 2019-2023, and the last row shows the
model with all the calibrated parameters in 2005-2009. To evaluate the
model, I use the marriage rate from the Japanese Census 2005 and the
cohort fertility rate (CFR) from the Human Fertility Database (HFD)
2005-2009, which are the targets in the baseline calibration.

\begin{table}

\caption{\label{tbl-past-decomp}Decomposition of the Marriage Rate and
CFR in 2005-2009}

\centering{

\centering
\begin{talltblr}[         %% tabularray outer open
entry=none,label=none,
note{}={Notes: The model is simulated with the calibrated parameters in 2005-2009. The mark $\checkmark$ means the values in 2005-2009 are used. The marriage rate is from the Japanese Census 2005 and 2020 and the cohort fertility rate (CFR) is from the Human Fertility Database (HFD) 2005-2009 and 2019-2023.},
]                     %% tabularray outer close
{                     %% tabularray inner open
width={0.85\linewidth},
colspec={X[0.3]X[0.05]X[0.05]X[0.05]X[0.1]X[0.1]X[0.1]X[0.1]},
hline{2}={8}{solid, black, 0.05em},
hline{3}={1-8}{solid, black, 0.05em},
hline{2}={5,7}{solid, black, 0.05em, l=-0.5},
hline{2}={6}{solid, black, 0.05em, r=-0.5},
hline{1}={1-8}{solid, black, 0.1em},
hline{8}={1-8}{solid, black, 0.1em},
cell{1}{1}={}{halign=c},
cell{1}{2}={}{halign=c},
cell{1}{3}={}{halign=c},
cell{1}{4}={}{halign=c},
cell{1}{5}={c=2}{halign=c, wd=0.2\linewidth},
cell{1}{6}={}{halign=c},
cell{1}{7}={c=2}{halign=c, wd=0.2\linewidth},
cell{1}{8}={}{halign=c},
}                     %% tabularray inner close
&  &  &  & Marriage Rate &  & Fertility Rate &  \\
& $\alpha_l$ & $\mu_f$ & $\theta$ & Model & Data & Model & Data \\
Baseline (2019-2023) &  &  &  & 0.839 & 0.836 & 1.631 & 1.454 \\
Leisure Technology & $\checkmark$ &  &  & 0.844 &   & 1.702 &   \\
Female Wage &  & $\checkmark$ &  & 0.839 &   & 1.630 &   \\
Social Norms &  &  & $\checkmark$ & 0.857 &   & 1.757 &   \\
All (2005-2009) & $\checkmark$ & $\checkmark$ & $\checkmark$ & 0.859 & 0.928 & 1.817 & 1.709 \\
\end{talltblr}

}

\end{table}%

Overall, the model with the calibrated parameters in 2005-2009 can
capture 21.1\% marriage rate and 73.1\% in the CFR. To study the role of
the three factors, I simulated the model with each calibrated parameter
in 2005-2009 and showed the results in the second to the fourth rows of
Table~\ref{tbl-past-decomp}. The second row shows that the leisure
technology growth has a negative impact on the marriage rate (by 4.9\%)
and CFR (by 28.0\%). The impact is especially large on CFR because the
leisure technology growth increases not only the value of being single
but also the relative value of children to leisure. The third row
suggests that female wage growth has almost no impact on the marriage
rate and CFR. This might be because the female wage growth increases
both the value of being single and married. The female values of singles
are improved by the income shock but the value of marriage is also
increased by the income shock, rise in her bargaining power, and the
affordability of children. The fourth row suggests that the smaller
value of \(\theta\), or weaker norms on gender roles, reduces the
marriage rate (by 19.2\%) and CFR (by 49.6\%). Given the higher wages of
men, the smaller \(\theta\) is a negative productivity shock for
domestic labor and reduces the value of marriage.

\section{Conclusion}\label{sec-concl}

In this paper, I build a model that can account for the marriage and
fertility decline in the last decade. The key ingredients of the model
are intra-household bargaining, leisure technology growth, and the
social norms on gender roles. The baseline model explains 21.1\% of the
marriage decline from 2005 to 2020 and 73.1\% of the fertility decline
from 2005-2009 to 2019-2023.

I also decompose the three drivers of the marriage and fertility
decline. The simulation results show that the leisure technology growth
has a significant impact on the fertility decline, which suggests that
the leisure technology growth has increased the value of being single
and also the child penalty on leisure for married couples.

\newpage{}

\appendix
\renewcommand{\thetable}{\Alph{section}.\arabic{table}}
\setcounter{table}{0}
\renewcommand{\thefigure}{\Alph{section}.\arabic{figure}}
\setcounter{figure}{0}

\begin{center}
{\huge\sffamily\bfseries Appendix}
\end{center}
\vspace{1em}

\section{Mathematical Derivations}\label{sec-app-math}

\subsection{Mass of the New Arrivals}\label{sec-app-mass}

\begin{proposition}[Mass of the New
Arrivals]\protect\hypertarget{prp-mass-arriavals}{}\label{prp-mass-arriavals}

In the equation Eq.~\ref{eq-dist-single}, the mass of new arrivals is
given by the following formula:

\[
D^N = \frac{\kappa}{3}.
\]

\end{proposition}

\begin{proof}
Define the mass of each \(a \in \{Y, M, O\}\) as \(D^{a}\). Given the
common transition probability \(\kappa\), the mass of new arrivals is
given by the following equation:

\[
\begin{aligned}
D^Y &= D^N + (1 - \kappa) D^Y \\
D^M &= \kappa D^Y + (1 - \kappa) D^M\\
D^O &= \kappa D^M + (1 - \kappa) D^O \\
1 &= D^Y + D^M + D^O.
\end{aligned}
\]

Solving the above equations, we get the mass of each group as follows:

\[
\begin{aligned}
D^N &= \frac{\kappa}{3} \\
D^Y &= \frac{1}{3} \\
D^M &= \frac{1}{3} \\
D^O &= \frac{1}{3}.
\end{aligned}
\]
\end{proof}

\subsection{Bargaining Power and Leisure Time}\label{sec-app-bargaining}

In Figure~\ref{fig-intra-household}, we have shown that the
intra-household gaps in earnings are positively correlated with leisure
time. This is consistent with the model economy, where the bargaining
power \(\rho\) is larger than 1.

\begin{proposition}[Positive Correlation in Intra-Household Gaps in
Wages and
Leisure]\protect\hypertarget{prp-bargaining}{}\label{prp-bargaining}

In the married couple's time allocation problem, the gaps in log wages
(\(\log w_m - \log w_f\)) and leisure time (\(\log l_m - \log l_f\)) are
positively correlated if and only if the bargaining power curvature
\(\rho\) is larger than 1. In addition, the gaps in leisure time are
increased by the existence of small children (\(N_0 > 0\)) if
\(\rho_2 > 0\).

\end{proposition}

\begin{proof}
From the first order condition of the utility function in
Eq.~\ref{eq-utility-married} with respect to \(l_m\) and \(l_f\), we
have the following equation:

\[
\begin{aligned}
(1 - \lambda) \alpha_l l_m^{-\gamma_l} &= \eta w_m \\
\lambda \alpha_l l_f^{-\gamma_l} &= \eta w_f
\end{aligned}
\]

where \(\eta\) is a Lagrange multiplier for the budget constraint. From
the above equations and
\(\lambda = \Lambda(w_m, w_f, N_0) = \frac{1}{1 + \exp\left(\rho_0 + \rho_1 (\log w_m - \log w_f) + \rho_2 \mathbb{1} \{N_0 > 0\} \right)}\),
we can derive the following equation:

\[
\log l_m - \log l_f = \frac{\rho_0}{\gamma_l} + \frac{\rho_1 - 1}{\gamma_l} (\log w_m - \log w_f) + \frac{\rho_2}{\gamma_l} \mathbb{1}\{N_0 > 0\}.
\]

Since the utility curvature parameter \(\gamma_l > 0\), the positive
correlation between the intra-household gaps in wages and leisure time
happens if and only if the bargaining power \(\rho_1\) is larger than 1.
Similarly, the existence of small children increases the gaps in leisure
time if \(\rho_2 > 0\).
\end{proof}

\subsection{\texorpdfstring{Domestic Labor Productivity
\(\theta\)}{Domestic Labor Productivity \textbackslash theta}}\label{sec-app-theta}

\begin{proposition}[Calibration of Domestic Labor
Productivity]\protect\hypertarget{prp-theta}{}\label{prp-theta}

The domestic labor productivity parameter \(\theta\) from
Eq.~\ref{eq-joint-domestic-labor} is fully characterized by the
first-order condition for the maximization problem in
Eq.~\ref{eq-utility-married}:

\[
\theta = \frac{w_f d_f^{1-\xi}}{w_m d_m^{1-\xi} + w_f d_f^{1-\xi}}.
\]

\end{proposition}

\begin{proof}
The cost minimization problem of joint domestic labor is given by

\[
\min_{d_m, d_f} w_m d_m + w_f d_f
\]

subject to

\[
D(d_m, d_f) = \left((1-\theta) d_m^{\xi} + \theta d_f^{\xi}\right)^{\frac{1}{\xi}} = \psi_0 + \psi_1 \mathbb{1}\{N_0 > 0\} + \psi_2 \mathbb{1}\{N_0 + N_1> 0\}.
\]

The first order conditions with respect to \(d_m\) and \(d_f\) are

\[
\begin{aligned}
w_m &= \eta \left((1-\theta) d_m^{\xi} + \theta d_f^{\xi}\right)^{\frac{1}{\xi} - 1}(1-\theta)d_m^{\xi-1}  \\
w_f &= \eta \left((1-\theta) d_m^{\xi} + \theta d_f^{\xi}\right)^{\frac{1}{\xi} - 1}\theta d_f^{\xi-1}.
\end{aligned}
\]

where \(\eta\) is a Lagrange multiplier for the cost minimization
problem. From the above equations, we can derive the following equation

\[
\frac{w_m}{w_f} = \frac{1-\theta}{\theta} \frac{d_m^{\xi-1}}{d_f^{\xi-1}},
\]

and we will get the equation of the proposition.
\end{proof}

\section{Data Description}\label{sec-app-data}

\subsection{Japanese Household Panel Survey (JHPS)}\label{sec-app-jhps}

The analysis is mostly based on the Japanese Household Panel Survey
(JHPS). The JHPS has been implemented annually since 2004 by the Panel
Data Research Center at Keio University and was originally named the
Keio Household Panel Survey (KHPS). The purpose of the KHPS is to
collect panel data on households and individuals reflecting the
population composition of society as a whole, as in the Panel Study of
Income Dynamics (PSID) in the U.S. and the European Community Household
Panel (ECHP) in Europe. KHPS started in 2004 with 4000 households and
7,000 individuals nationwide and added a cohort of about 1400 households
and 2500 individuals from 2007 to compensate for sample dropout. In
2009, the Panel Data Research Center at Keio University started a new
survey, the JHPS, targeting 4000 male and female subjects nationwide in
parallel with the KHPS. The JHPS collects data focused on education and
health/healthcare in addition to economic status and employment status.
In 2014, the KHPS was merged with the JHPS.

As described in Section~\ref{sec-data-samples}, the sample is restricted
to people aged 25-54 in the period 2005-2023. The sample of singles is
restricted to those who have a job, a positive leisure time, and have no
children. The sample of married couples is also restricted to those who
have a positive leisure time, however, it includes the case of
non-working individuals. Table~\ref{tbl-sum-single} and
Table~\ref{tbl-sum-couple} show the summary statistics of singles and
married couples, respectively.

\begin{table}

\caption{\label{tbl-sum-single}Summary Statistics of Singles in
JHPS2005-2023}

\centering{

\centering
\begin{talltblr}[         %% tabularray outer open
entry=none,label=none,
note{}={Notes: The table shows the summary statistics of singles aged 25-54 in the period 2005-2023. The sample of singles is restricted to those who have a job, a positive leisure time, and have no children.},
]                     %% tabularray outer close
{                     %% tabularray inner open
width={0.9\linewidth},
colspec={X[0.3]X[0.15]X[0.15]X[0.15]X[0.15]},
hline{2}={5}{solid, black, 0.05em},
hline{3}={1-5}{solid, black, 0.05em},
hline{2}={2,4}{solid, black, 0.05em, l=-0.5},
hline{2}={3}{solid, black, 0.05em, r=-0.5},
hline{1}={1-5}{solid, black, 0.1em},
hline{7}={1-5}{solid, black, 0.1em},
cell{1}{1}={}{halign=c},
cell{1}{2}={c=2}{halign=c, wd=0.3\linewidth},
cell{1}{3}={}{halign=c},
cell{1}{4}={c=2}{halign=c, wd=0.3\linewidth},
cell{1}{5}={}{halign=c},
}                     %% tabularray inner close
& Men (N = 1199) &  & Women (N = 1119) &  \\
& Mean & Std. Dev. & Mean & Std. Dev. \\
Hours worked (per week) & 48.7 & 17.1 & 41.7 & 15.8 \\
Leisure (per week) & 59.7 & 17.6 & 60.8 & 16.9 \\
Domestic labor (per week) & 3.5 & 5.4 & 9.5 & 10.0 \\
Hourly wage (JPY) & 1702.1 & 1640.4 & 1400.8 & 2034.2 \\
\end{talltblr}

}

\end{table}%

\begin{table}

\caption{\label{tbl-sum-couple}Summary Statistics of Married Couples in
JHPS2005-2023}

\centering{

\centering
\begin{talltblr}[         %% tabularray outer open
entry=none,label=none,
note{}={Notes: The table shows the summary statistics of married couples aged 25-54 in the period 2005-2023. The sample of married couples is restricted to those who have a positive leisure time and have up to three children.},
]                     %% tabularray outer close
{                     %% tabularray inner open
width={0.9\linewidth},
colspec={X[0.3]X[0.15]X[0.15]X[0.15]X[0.15]},
hline{2}={5}{solid, black, 0.05em},
hline{3}={1-5}{solid, black, 0.05em},
hline{2}={2,4}{solid, black, 0.05em, l=-0.5},
hline{2}={3}{solid, black, 0.05em, r=-0.5},
hline{1}={1-5}{solid, black, 0.1em},
hline{7}={1-5}{solid, black, 0.1em},
cell{1}{1}={}{halign=c},
cell{1}{2}={c=2}{halign=c, wd=0.3\linewidth},
cell{1}{3}={}{halign=c},
cell{1}{4}={c=2}{halign=c, wd=0.3\linewidth},
cell{1}{5}={}{halign=c},
}                     %% tabularray inner close
& Men (N = 1981) &  & Women (N = 2034) &  \\
& Mean & Std. Dev. & Mean & Std. Dev. \\
Hours worked (per week) & 49.1 & 20.6 & 22.3 & 19.0 \\
Leisure (per week) & 58.4 & 21.3 & 52.9 & 22.7 \\
Domestic labor (per week) & 4.5 & 7.3 & 36.8 & 23.5 \\
Hourly wage (JPY) & 2382.6 & 2670.9 & 1263.3 & 1446.7 \\
\end{talltblr}

}

\end{table}%

\subsection{Marriage Rate by Earnings Decile}\label{sec-app-ess}

In Figure~\ref{fig-marriage-earn}, I show the marriage rate by earnings
deciles from the Employment Status Survey 2022. In this survey, the
number of people by earnings, age, and marital status is reported.
Figure~\ref{fig-earn-ess} shows the original data of the number of
people aged between 45 and 54 by earnings and marital status.

\begin{figure}

\centering{

\includegraphics[width=0.9\linewidth,height=\textheight,keepaspectratio]{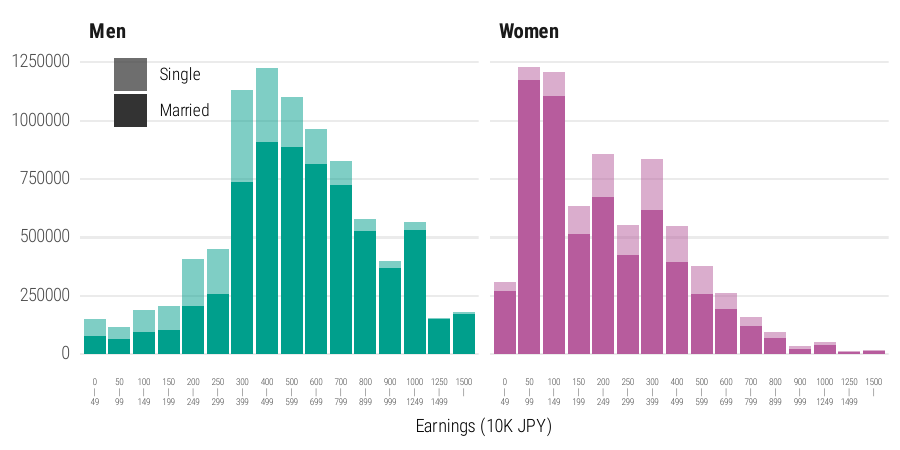}

}

\caption{\label{fig-earn-ess}\textbf{Number of People Aged between 45
and 54 by Earnings and Marital Status}. The data is from the Employment
Status Survey 2022.}

\end{figure}%

To estimate the marriage rate by earnings deciles, I compute the
cumulative density of the earnings distribution by marital status. Since
the original data reports only the 16 categories of earnings, I estimate
the cumulative density by linear interpolation. Figure~\ref{fig-cdf-ess}
shows the estimated cumulative density of earnings by marital status.

\begin{figure}

\centering{

\includegraphics[width=0.9\linewidth,height=\textheight,keepaspectratio]{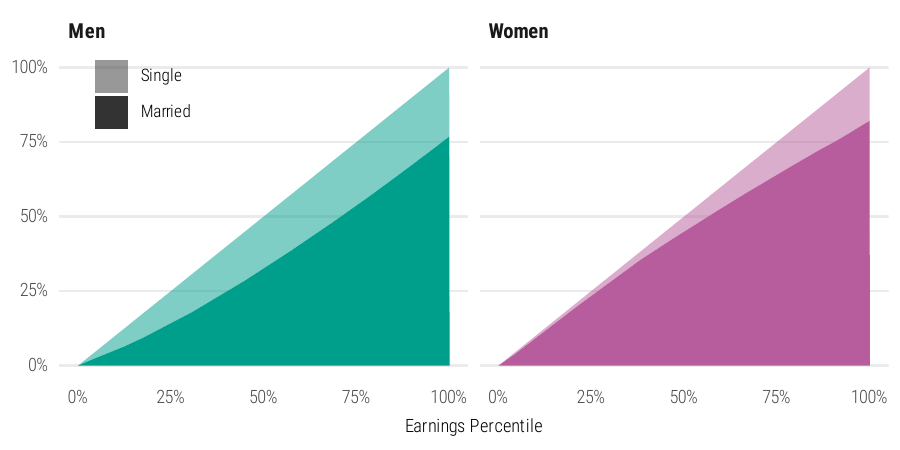}

}

\caption{\label{fig-cdf-ess}\textbf{Cumulative Density of Earnings by
Marital Status}. The data is from the Employment Status Survey 2022.}

\end{figure}%

The marriage rate by earnings deciles is defined as the ratio of the
mass of married individuals to the cumulative density of single and
married individuals. Since the denominator is 10\% by definition, the
marriage rate by earnings decile is given by

\[
\text{Marriage Rate}_n = \frac{\hat{M}(0.1n) - \hat{M}(0.1(n-1))}{0.1},
\]

where \(\hat{M}(q)\) is the estimated cumulative density of married
individuals at the \(q\)-th percentile of earnings.

\section{Supplemental Figures and Tables}\label{sec-app-figure}

\begin{figure}

\centering{

\includegraphics[width=0.9\linewidth,height=\textheight,keepaspectratio]{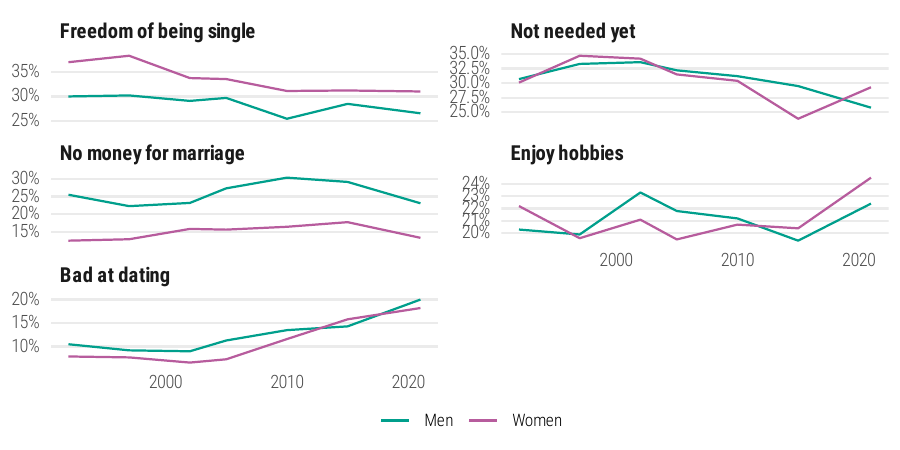}

}

\caption{\label{fig-reason-single-ts}\textbf{Time Series of Reasons Why
Do not Get Married}. The figure shows the time trends of the share of
the main reasons why do not get married. The data is from the National
Fertility Survey from 1992 to 2021 and the sample is restricted to the
age group 25-34.}

\end{figure}%

\begin{figure}

\centering{

\includegraphics[width=0.9\linewidth,height=\textheight,keepaspectratio]{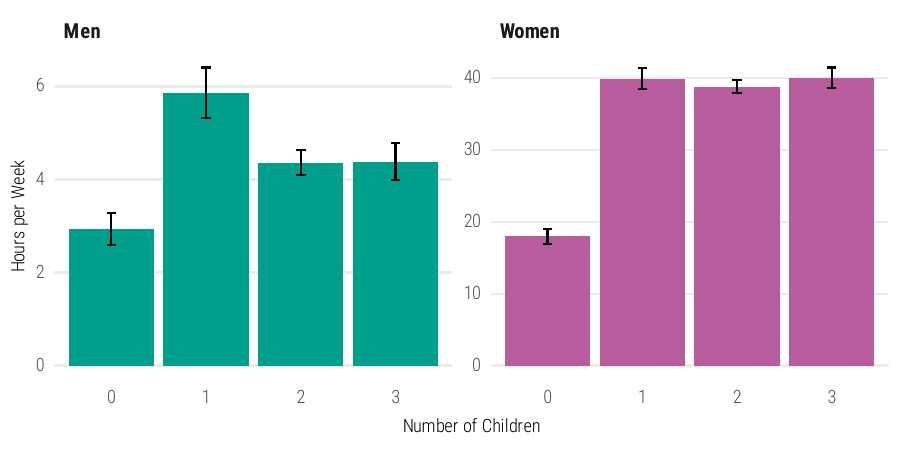}

}

\caption{\label{fig-hours-domestic}\textbf{Domestic Labor Hours by
Number of Children}. The figure shows the average hours of domestic
labor by the number of children. The error bar shows the 95\% confidence
interval with the standard error clustered by household. The data is
from the Japanese Household Panel Survey from 2005 to 2023 and the
sample is restricted to married individuals with age 25-54.}

\end{figure}%

\begin{table}

\caption{\label{tbl-smm-fit}Model and Data Moments}

\centering{

\centering
\begin{tblr}[         %% tabularray outer open
]                     %% tabularray outer close
{                     %% tabularray inner open
colspec={Q[]Q[]Q[]Q[]},
hline{2}={1-4}{solid, black, 0.05em},
hline{1}={1-4}{solid, black, 0.1em},
hline{22}={1-4}{solid, black, 0.1em},
}                     %% tabularray inner close
& Data & Model & Source \\
Single $l_m$ & 0.560 & 0.551 & JHPS2019-2023 \\
Single $l_f$ & 0.571 & 0.502 & JHPS2019-2023 \\
Married $l_m$, without children & 0.536 & 0.532 & JHPS2019-2023 \\
Married $l_f$, without children & 0.573 & 0.575 & JHPS2019-2023 \\
Married $l_m$, with small children (< 7 years old) & 0.445 & 0.504 & JHPS2019-2023 \\
Married $l_f$, with small children (< 7 years old) & 0.304 & 0.301 & JHPS2019-2023 \\
Married $l_m$, with older children (>= 7 years old) & 0.539 & 0.490 & JHPS2019-2023 \\
Married $l_f$, with older children (>= 7 years old) & 0.496 & 0.529 & JHPS2019-2023 \\
Married $d_m$, without children & 0.035 & 0.035 & JHPS2019-2023 \\
Married $d_f$, without children & 0.150 & 0.154 & JHPS2019-2023 \\
Married $d_m$, with small children (< 7 years old) & 0.124 & 0.104 & JHPS2019-2023 \\
Married $d_f$, with small children (< 7 years old) & 0.489 & 0.550 & JHPS2019-2023 \\
Married $d_m$, with older children (>= 7 years old) & 0.040 & 0.045 & JHPS2019-2023 \\
Married $d_f$, with older children (>= 7 years old) & 0.287 & 0.227 & JHPS2019-2023 \\
Share of women with one child & 0.197 & 0.184 & HFD2019-2023 \\
Share of women with two children & 0.361 & 0.349 & HFD2019-2023 \\
Share of women with three or more children & 0.165 & 0.162 & HFD2019-2023 \\
Mean difference in single's $\log w_m$ and $\log w_f$ & 0.148 & 0.147 & JHPS2019-2023 \\
S.D. of single's $\log w_f$ & 0.766 & 0.804 & JHPS2019-2023 \\
Share of never-married women & 0.164 & 0.161 & Census2020 \\
\end{tblr}

}

\end{table}%

\newpage{}

\section*{References}\label{references}
\addcontentsline{toc}{section}{References}

\protect\phantomsection\label{refs}

\end{document}